\documentclass{theoretics}

\title{Isomorphism for Tournaments of Small Twin Width}

\ThCSauthor{Martin Grohe}{grohe@informatik.rwth-aachen.de}[0000-0002-0292-9142]
\ThCSaffil{RWTH Aachen University, Aachen, Germany}
\ThCSauthor{Daniel Neuen}{daniel.neuen@tu-dresden.de}[0000-0002-4940-0318]
\ThCSaffil{TU Dresden, Dresden, Germany}

\ThCSthanks{The first author is funded by the European Union (ERC, SymSim, 101054974).
Views and opinions expressed are however those of the author(s) only and do not necessarily reflect those of the European Union or the European Research Council.
Neither the European Union nor the granting authority can be held responsible for them. A preliminary version of this article appeared at ICALP 2024 \cite{GroheN24}.}

\ThCSshortnames{M.\ Grohe, D.\ Neuen}  
\ThCSshorttitle{Isomorphism for Tournaments of Small Twin Width}

\ThCSyear{2026}
\ThCSarticlenum{4}
\ThCSreceived{Aug 14, 2024}
\ThCSrevised{Nov 12, 2025}
\ThCSaccepted{Feb 8, 2026}
\ThCSpublished{Feb 21, 2026}
\ThCSdoicreatedtrue
\ThCSkeywords{tournament isomorphism, twin width, fixed-parameter tractability, Weisfeiler-Leman algorithm}

\usepackage[T1]{fontenc}
\usepackage{amstext, mathtools}
\usepackage{pgfplots}
\usepackage{tikz}

\usetikzlibrary{patterns,arrows,decorations.pathreplacing,decorations.pathmorphing,calc}

\tikzstyle{vertex}=[draw,circle,fill=white,minimum size=6pt,inner sep=0pt]
\tikzstyle{smallvertex}=[draw,circle,fill=white,minimum size=5pt,inner sep=0pt]

\newcommand{\CC}{{\mathcal C}}

\newcommand{\CE}{{\mathcal E}}

\newcommand{\CH}{{\mathcal H}}

\newcommand{\CL}{{\mathcal L}}
\newcommand{\CM}{{\mathcal M}}

\newcommand{\CP}{{\mathcal P}}
\newcommand{\CQ}{{\mathcal Q}}

\newcommand{\CS}{{\mathcal S}}
\newcommand{\CT}{{\mathcal T}}

\newcommand{\CW}{{\mathcal W}}

\newcommand{\NN}{{\mathbb N}}
\newcommand{\ZZ}{{\mathbb Z}}

\newcommand{\bigmid}{\;\big|\;}
\newcommand{\Bigmid}{\;\Big|\;}

\newcommand{\WL}[2]{\chi^{#1,#2}}
\newcommand{\WLit}[3]{\chi_{(#2)}^{#1,#3}}

\newcommand{\red}{{\sf red}}
\newcommand{\Ered}{E_{\sf red}}

\DeclareMathOperator{\md}{md}

\DeclareMathOperator{\Aut}{Aut}
\DeclareMathOperator{\Iso}{Iso}

\DeclareMathOperator{\Sym}{Sym}

\DeclareMathOperator{\CopRob}{CopRob}
\DeclareMathOperator{\BP}{BP}

\DeclareMathOperator{\tw}{tw}

\DeclareMathOperator{\width}{width}
\DeclareMathOperator{\dtw}{dtw}
\DeclareMathOperator{\tww}{tww}
\DeclareMathOperator{\dpw}{dpw}
\DeclareMathOperator{\ctw}{ctw}

\newcommand{\CFI}{{\sf CFI}}

\newcommand{\blue}{{\sf blue}}
\newcommand{\green}{{\sf green}}

\pgfplotsset{compat=1.18} 

\begin{document}

\maketitle

\begin{abstract}
 We prove that isomorphism of tournaments of twin width at most $k$ can be decided in time $k^{O(\log k)}n^{O(1)}$.
 This implies that the isomorphism problem for classes of tournaments of bounded or moderately growing twin width is in polynomial time.
 By comparison, there are classes of undirected graphs of bounded twin width that are isomorphism complete, that is, the isomorphism problem for the classes is as hard as the general graph isomorphism problem.
 Twin width is a graph parameter that has been introduced only recently (Bonnet et al., J.\ ACM 2022), but has received a lot of attention in structural graph theory since then.
 On directed graphs, it is functionally smaller than clique width.
 We prove that on tournaments (but not on general directed graphs) it is also functionally smaller than directed tree width (and thus, the same also holds for cut width and directed path width).
 Hence, our result implies that tournament isomorphism testing is also fixed-parameter tractable when parameterized by any of these parameters.

 Our isomorphism algorithm heavily employs group-theoretic techniques.
 This seems to be necessary: as a second main result, we show that the combinatorial Weisfeiler-Leman algorithm does not decide isomorphism of tournaments of twin width at most 35 if its dimension is $o(n)$.
 (Throughout this abstract, $n$ is the order of the input graphs.)
\end{abstract}

\section{Introduction}

The tournament isomorphism problem (TI) was recognized as a particularly interesting special case of the graph isomorphism problem (GI) early-on.
Already in 1983, Babai and Luks~\cite{BabaiL83} proved that TI is solvable in time $n^{O(\log n)}$;
it took 33 more years for Babai \cite{Babai16} to prove that the general GI is in quasi-polynomial time.
An important fact that makes TI more accessible than GI is that tournaments always have solvable automorphism groups.
This is a consequence of the observation that the automorphism groups of tournaments have odd order and the famous Feit-Thompson Theorem~\cite{FeitT63} stating that all groups of odd order are solvable.
However, even Babai's powerful new machinery did not help us to improve the upper bound for TI, as one might have hoped.
But TI is not only special from a group-theoretic perspective.
Another remarkable result, due to Schweitzer~\cite{Schweitzer17}, states that TI reduces to the problem of deciding whether a tournament has a nontrivial automorphism; the so-called \emph{rigidity problem}.
It is an open question whether the same holds for general graphs.

While there is an extensive literature on GI restricted to classes of graphs (see \cite{GroheN21,Neuen26} for recent surveys), remarkably little is known for restrictions of TI.
Ponomarenko \cite{Ponomarenko92} proved that TI is in polynomial time for tournaments whose automorphism group contains a regular cyclic subgroup,
and recently Arvind, Ponomarenko, and Ryabov~\cite{ArvindPR25} proved that TI is in polynomial time for edge-colored tournaments where at least one edge color induces a (strongly) connected spanning subgraph of bounded degree (even fixed-parameter tractable when parameterized by the out-degree).
While both of these results are very interesting from a technical perspective, they consider classes of tournaments that would hardly be called natural from a graph-theoretic point of view.
Natural graph parameters that have played a central role in the structural theory of tournaments developed by Chudnovsky, Seymour and others \cite{ChudnovskyFS12,ChudnovskyKLST18,ChudnovskySS19,ChudnovskyS11,FominP19} are \emph{cut width} and \emph{path width}.
The more recent theory of structural sparsity \cite{GajarskyKNMPST20,GajarskyPT22,NesetrilO16,NesetrilMS22} highlights \emph{clique width} and \emph{twin width}.
Here twin width is the key parameter.
Not only is it functionally smaller than the other parameters, which means that if cut width, path width, or clique width is bounded, then twin width is bounded as well, it is also known \cite{GenietT26} that a class of tournaments has bounded twin width if and only if it has a property known as \emph{monadic dependence (NIP)}.
Dependence is a key property studied in classical model theory.
A class of graphs is monadically dependent if and only if all set systems definable in this class by a first-order transduction have bounded VC dimension.
This property seems to characterize precisely the graph classes that are regarded as structurally sparse.
Since twin width of graphs and binary relational structures has been introduced in \cite{BonnetKTW22}, it received a lot of attention in algorithmic structural graph theory
\cite{BonnetGKTW22,BonnetGKTW24,BonnetGMSTT24,BonnetGMT23,BonnetKRT22,BonnetKRTW22,GajarskyPT22,GanianPSSS22,GenietT26,HlinenyJ25,Thomasse22}.
(We defer the somewhat unwieldy definition of twin width to Section~\ref{sec:tww}.)
Our main result states that tournament isomorphism is fixed-parameter tractable when parameterized by twin width.

\begin{theorem}
 \label{thm:tww-isomorphism-intro}
 The isomorphism problem for tournaments of twin width at most $k$ can be solved in time $k^{O(\log k)}\cdot n^{O(1)}$.
\end{theorem}

Interestingly, isomorphism testing for undirected graphs of bounded twin width at most $4$ is as hard as the general GI.
This follows easily from the fact that a $\lceil 2\log n\rceil$-subdivision of every graph with $n$ vertices has twin width at most $4$ \cite{BergeBD22}.
Once more, this demonstrates the special role of the tournament isomorphism problem, though here the reason is not group-theoretic, but purely combinatorial.

Note that the dependence on the twin width $k$ of the algorithm in Theorem~\ref{thm:tww-isomorphism-intro} is subexponential, so our result implies that TI is in polynomial time even for tournaments of twin width $2^{O(\sqrt{\log n})}$.
Since twin width is functionally smaller than clique width, our result implies that TI is also fixed-parameter tractable when parameterized by clique width.
Additionally, we prove (Corollary~\ref{cor:tww-vs-dtw}) that the twin width of a tournament is functionally smaller than its \emph{directed tree width}, a graph parameter originally introduced in \cite{JohnsonRST01}.
Since the directed tree width of every directed graph is smaller than its cut width or directed path width, the same also holds for these two parameters.
Hence, TI is fixed-parameter tractable also when parameterized by directed tree width, directed path width or cut width.
To the best of our knowledge, this was not known for any of these parameters.
The fact that twin width is functionally smaller than directed tree width, directed path width and cut width on tournaments is interesting in its own right, because this result does not extend to general directed graphs (for any of the three parameters).

Our proof of Theorem \ref{thm:tww-isomorphism-intro} heavily relies on group-theoretic techniques.
In a nutshell, we show that bounded twin width allows us to cover a tournament by a sequence of directed graphs that have a property resembling bounded degree sufficiently closely to apply a group-theoretic machinery going back to Luks \cite{Luks82} and developed to great depth since then (see, e.g.,~\cite{Babai16,BabaiL83,GroheNS23,Miller83,Neuen22}).
Specifically, we generalize arguments that have been introduced by Arvind et al.~\cite{ArvindPR25} for TI on edge-colored tournaments where at least one edge color induces a spanning subgraph of bounded out-degree.

Yet one may wonder if this heavy machinery is even needed to prove our theorem, in particular in view of the fact that on many natural graph classes,
including, for example, undirected graphs of bounded clique width \cite{GroheN23}, the purely combinatorial Weisfeiler-Leman algorithm is sufficient to decide isomorphism (see, e.g.,~\cite{Grohe17,Kiefer20}).
We prove that this is not the case for tournaments of bounded twin width.

\begin{theorem}
 \label{thm:wl-tournament-tww}
 For every $k \geq 2$ there are non-isomorphic tournaments $T_k$ and $T_k'$ of order $|V(T_k)| = |V(T_k')| = O(k)$ and twin width at most $35$ that are not distinguished by the $k$-dimensional Weisfeiler-Leman algorithm.
\end{theorem}

We remark that it was known before that the Weisfeiler-Leman algorithm fails to decide tournament isomorphism.
Indeed, Dawar and Kopczynski (unpublished) proved that for every $k \geq 2$ there are non-isomorphic tournaments $U_k$ and $U_k'$ of order $|V(U_k)| = |V(U_k')| = O(k)$ that are not distinguished by the $k$-dimensional Weisfeiler-Leman algorithm.
Theorem \ref{thm:wl-tournament-tww} strengthens this result by constructing tournaments where additionally the twin width is bounded by a fixed constant.

The paper is organized as follows.
After introducing the necessary preliminaries in Section~\ref{sec:preliminaries}, Theorem \ref{thm:tww-isomorphism-intro} is proved in Sections \ref{sec:partitions} and \ref{sec:alg}.
First, we give our main combinatorial arguments in Section~\ref{sec:partitions}.
After that, the mainly group-theoretic isomorphism algorithm of Theorem~\ref{thm:tww-isomorphism-intro} is presented in Section~\ref{sec:alg}.
Theorem~\ref{thm:wl-tournament-tww} is proved in Section~\ref{sec:wl}.
Finally, in Section~\ref{sec:widths} we compare twin width to other width measures for directed graphs.

\section{Preliminaries}
\label{sec:preliminaries}

\subsection{Graphs}

Graphs in this paper are usually directed.
We often emphasize this by calling them ``digraphs''.
However, when we make general statements about graphs, this refers to directed graphs and includes undirected graphs a special case (directed graphs with a symmetric edge relation).
We denote the vertex set of a graph $G$ by $V(G)$ and the edge relation by $E(G)$.
The vertex set $V(G)$ is always finite and non-empty.
The edge relation is always anti-reflexive, that is, graphs are loop-free, and there are no parallel edges.
For a digraph $G$ and a vertex $v \in V(G)$, we denote the set of out-neighbors and in-neighbors of $v$ by $N_+(v)$ and $N_-(v)$, respectively.
Also, the \emph{out-degree} and \emph{in-degree} of $v$ are denoted by $\deg_+(v) \coloneqq |N_+(v)|$ and $\deg_-(v) \coloneqq |N_-(v)|$, respectively.
Furthermore, $E_+(v)$ and $E_-(v)$ denote the set of outgoing and incoming edges into $v$, respectively.
For $X\subseteq V(G)$, we write $G[X]$ to denote the subgraph of $G$ induced on $X$.
For two sets $X,Y \subseteq V(G)$ we write $E_G(X,Y) \coloneqq \{(v,w) \in E(G) \mid v \in X, w \in Y\}$ to denote the set of directed edges from $X$ to $Y$.

Let $G$ be an undirected graph.
A directed graph $\vec{G}$ is an \emph{orientation} of $G$ if, for every undirected edge $\{v,w\} \in E(G)$,
exactly one of $(v,w)$ and $(w,v)$ is an edge of $\vec{G}$, and there are no other edges present in $\vec{G}$.
A \emph{tournament} is an orientation of a complete graph.

A tournament $T$ is \emph{regular} if $\deg_+(v) = \deg_+(w)$ for all $v,w \in V(T)$.
In this case, $\deg_+(v) = \deg_-(v) = \frac{|V(G)|-1}{2}$ for all $v \in V(T)$.
This implies that every regular tournament has an odd number of vertices.

Let $G_1,G_2$ be two graphs.
An \emph{isomorphism} from $G_1$ to $G_2$ is a bijection $\varphi\colon V(G_1) \to V(G_2)$ such that $(v,w) \in E(G_1)$ if and only if $(\varphi(v),\varphi(w)) \in E(G_2)$ for all $v,w \in V(G_1)$.
We write $\varphi\colon G_1) \cong G_2$ to denote that $\varphi$ is an isomorphism from $G_1$ to $G_2$.
Also, $\Iso(G_1,G_2)$ denotes the set of all isomorphisms from $G_1$ to $G_2$.
The graphs $G_1$ and $G_2$ are \emph{isomorphic} if $\Iso(G_1,G_2) \neq \emptyset$.
The \emph{automorphism group} of $G_1$ is $\Aut(G_1) \coloneqq \Iso(G_1,G_1)$.

An \emph{arc coloring} of a digraph $G$ is a mapping $\lambda\colon (E(G) \cup \{(v,v)\mid v\in V(G)\}) \to C$ for some set $C$ of ``colors''.
An \emph{arc-colored graph} is a triple $G=(V,E,\lambda)$, where $(V,E)$ is a graph an $\lambda$ an arc coloring of $(V,E)$.
Isomorphisms between arc-colored graphs are required to preserve the coloring.

\subsection{Partitions and Colorings}

Let $S$ be a finite set.
A \emph{partition} of $S$ is a set $\CP\subseteq 2^S$ whose elements we refer to as \emph{parts},
such that any two parts are mutually disjoint, and the union of all parts is $S$.
A partition $\CP$ \emph{refines} another partition $\CQ$, denoted by $\CP \preceq \CQ$, if for all $P \in \CP$ there is some $Q \in \CQ$ such that $P \subseteq Q$.
We say a partition $\CP$ is \emph{trivial} if $|\CP| = 1$, which means that the only part is $S$, and it is \emph{discrete} if $|P| = 1$ for all $P \in \CP$.

Every mapping $\chi\colon S\to C$, for some set $C$, induces a partition $\CP_\chi$ of $S$ into the sets $\chi^{-1}(c)$ for all $c$ in the range of $\chi$.
In this context, we think of $\chi$ as a ``coloring'' of $S$, the elements $c\in C$ as ``colors'', and the parts $\chi^{-1}(c)$ of the partition as ``color classes''.
If $\chi'\colon S\to C'$ is another coloring, then we say that $\chi$ \emph{refines} $\chi'$, denoted by $\chi \preceq \chi'$, if $\CP_\chi \preceq \CP_{\chi'}$.
The colorings are \emph{equivalent} (we write $\chi \equiv \chi'$) if $\chi \preceq \chi'$ and $\chi \preceq \chi'$, i.e., $\CP_\chi = \CP_{\chi'}$.

\subsection{Twin Width}
\label{sec:tww}

Twin width \cite{BonnetKTW22} is defined for binary relational structures, which in this paper are mostly directed graphs.
We need one distinguished binary relation symbol $\Ered$ that plays a special role in the definition of twin width.
Following \cite{BonnetKTW22}, we refer to elements of $\Ered$ as \emph{red edges}.
For every structure $A$, we assume the relation $\Ered(A)$ to be symmetric and anti-reflexive, that is, the edge relation of an undirected graph, and we refer to the maximum degree of this graph as the \emph{red degree} of $A$.
If $\Ered(A)$ is not explicitly defined, we assume $\Ered(A)=\emptyset$ (and the red degree of $A$ is $0$).

Let $A = (V(A),R_1(A),\dots,R_k(A))$ be a binary relational structure, where $V(A)$ is a non-empty finite vertex set and $R_i(A) \subseteq (V(A))^2$ are binary relations on $V(A)$ (possibly, $R_i=\Ered$ for some $i \in [k]$).
We call a pair $(X,Y)$ of disjoint subsets of $V(A)$ \emph{homogeneous} if for all $x,x'\in X$, and all $y,y'\in Y$ it holds that
\begin{enumerate}[label = (\roman*)]
 \item $(x,y)\in R_i(A) \Leftrightarrow (x',y') \in R_i(A)$ and $(y,x) \in R_i(A) \Leftrightarrow (y',x') \in R_i(A)$ for all $i \in [k]$, and
 \item $(x,y) \notin \Ered(A)$ and $(y,x) \notin \Ered(A)$.
\end{enumerate}
For a partition $\CP$ of $V(A)$, we define $A/\CP$ to be the structure with vertex set $V(A/\CP) \coloneqq \CP$ and relations
\[R_i(A/\CP) \coloneqq \big\{(X,Y) \in \CP^2 \bigmid (X,Y) \text{ is homogeneous and } X \times Y\subseteq R_i(A)\big\}\]
for all $R_i \neq \Ered$, and
\[\Ered(A/\CP) \coloneqq \big\{(X,Y) \in \CP^2 \bigmid (X,Y) \text{ is not homogeneous and } X \neq Y\big\}.\]
A \emph{contraction sequence for $A$} is a sequence of partitions
$\CP_1,\dots,\CP_n$ of $V(A)$ such that $\CP_1 = \{\{v\} \mid v \in V(A)\}$ is the discrete partition, $\CP_n = \{V(A)\}$ is the trivial partition,
and for every $i \in [n-1]$ the partition $\CP_{i+1}$ is obtained from $\CP_i$ by merging two parts, i.e.,
there are distinct $P,P' \in \CP_i$ such that $\CP_{i+1} = \{P \cup P'\} \cup (\CP_{i} \setminus \{P,P'\})$.
The \emph{width} of a contraction sequence $\CP_1,\dots,\CP_n$ of $A$ is the minimum $k$ such that for every $i \in [n]$ the structure $A/\CP_i$ has red degree at most $k$.
The \emph{twin width} of $A$, denoted by $\tww(A)$, is the minimum $k \geq 0$ such that $A$ has a contraction sequence of width $k$.

Note that red edges are introduced as we contract parts of the partitions.
However, the structure $A$ we start with may already have red edges, which then have direct impact on its twin width.
In particular, the twin width of a graph $G$ may be smaller than the twin width of the structure $G^{\red}$ obtained from $G$ by coloring all edges red.
This fact is used later.

We also remark that for our isomorphism algorithms, we never have to compute a contraction sequence of minimum width or the twin width.

We state two simple lemmas on basic properties of twin width.

\begin{lemma}[\cite{BonnetKTW22}]
 \label{lem:tww-hereditary}
 Let ${A}$ be a binary relational structure and $X \subseteq V({A})$.
 Then $\tww({A}[X]) \leq \tww({A})$.
\end{lemma}

\begin{lemma}
 \label{lem:twin-width-order}
 Let ${A}$ be a structure over the vocabulary $\tau$.
 Then there is a linear order $<$ on $V({A})$ such that $\tww({A},<) = \tww({A})$.
\end{lemma}

\begin{proof}
 Let $\CP_1,\dots,\CP_n$ be a contraction sequence for ${A}$.
 Note that for every $i \in [n]$ the partition $\CP_i$ has exactly $n+1-i$ parts $P_{i,1},\dots,P_{i,n+1-i}$.
 We may choose the indices in such a way that if $\CP_{i}$ is obtained from $\CP_{i-1}$ by merging the parts $P_{i-1,j_1}$ and $P_{i-1,j_2}$ to $P_{i,j}=P_{i-1,j_1}\cup P_{i-1,j_2}$ then $j_1=j,j_2=j+1$, $P_{i,k}=P_{i-1,k}$ for $k<j$ and $P_{i,k}=P_{i-1,k+1}$ for $k>j$.
 Let $<$ be the linear order induced by the partition $\CP_1$ (whose parts have size $1$).
 Then for all $i \in [n]$, the parts of $\CP_i$ are consecutive intervals of $<$, which means that all pairs $(P_{i,j},P_{i,k})$ of distinct parts of $\CP_i$ are homogeneous for the relation $<$.
 This implies $\tww({A},<) = \tww({A})$.
\end{proof}

\subsection{Weisfeiler-Leman}
\label{sec:wl-def}

In this section, we describe the $k$-dimensional Weisfeiler-Leman algorithm ($k$-WL).
The algorithm has been originally introduced in its $2$-dimensional form by Weisfeiler and Leman \cite{WeisfeilerL68} (see also \cite{Weisfeiler76}).
The $k$-dimensional version, coloring $k$\nobreakdash-tuples, was introduced later by Babai and Mathon (see \cite{CaiFI92}).

Fix $k \geq 2$, and let $G$ be a graph.
For $i \geq 0$, we describe the coloring $\WLit{k}{i}{G}$ of $(V(G))^k$ computed in the $i$-th iteration of $k$-WL.
For $i = 0$, each tuple is colored with the isomorphism type of the underlying ordered induced subgraph.
So if $H$ is another graph and $\bar v = (v_1,\dots,v_k) \in (V(G))^k$, $\bar w = (w_1,\dots,w_k) \in (V(H))^k$,
then $\WLit{k}{0}{G}(\bar v) = \WLit{k}{0}{H}(\bar w)$ if and only if, for all $i,j \in [k]$, it holds that $v_i = v_j \Leftrightarrow w_i = w_j$ and $(v_i,v_j) \in E(G) \Leftrightarrow (w_i,w_j) \in E(H)$.
If $G$ and $H$ are arc-colored, then the colors are also taken into account.

Now let $i \geq 0$.
For $\bar v = (v_1,\dots,v_k)$ we define
\[\WLit{k}{i+1}{G}(\bar v) \coloneqq \Big(\WLit{k}{i}{G}(\bar v), \CM_{i}(\bar v)\Big)\]
where
\[\CM_i(\bar v) \coloneqq \Big\{\!\Big\{ \big(\WLit{k}{i}{G}(\bar v[w/1]),\dots,\WLit{k}{i}{G}(\bar v[w/k])\big) \Bigmid w \in V(G) \Big\}\!\Big\}\]
and $\bar v[w/i] \coloneqq (v_1,\dots,v_{i-1},w,v_{i+1},\dots,v_k)$ is the tuple obtained from $\bar v$ b replacing the $i$-th entry by $w$ (and $\{\!\{\dots\}\!\}$ denotes a multiset).

Clearly, $\WLit{k}{i+1}{G} \preceq \WLit{k}{i}{G}$ for all $i \geq 0$.
So there is a unique minimal $i_\infty \geq 0$ such that $\WLit{k}{i_\infty+1}{G} \equiv \WLit{k}{i_\infty}{G}$ and we write $\WL{k}{G} \coloneqq \WLit{k}{i_\infty}{G}$ to denote the corresponding coloring.

The $k$-dimensional Weisfeiler-Leman algorithm takes as input a (possibly colored) graph $G$ and outputs (a coloring that is equivalent to) $\WL{k}{G}$.
This can be done in time $O(k^2n^{k+1} \log n)$ \cite{ImmermanL90}.

Let $H$ be a second graph.
The $k$-dimensional Weisfeiler-Leman algorithm \emph{distinguishes} $G$ and $H$ if there is a color $c \in C$ such that
\[\Big|\Big\{ \bar v \in (V(G))^k \Bigmid \WL{k}{G}(\bar v) = c \Big\}\Big| \neq \Big|\Big\{ \bar w \in (V(H))^k \Bigmid \WL{k}{H}(\bar w) = c \Big\}\Big|.\]
We write $G \simeq_k H$ to denote that $k$-WL does not distinguish between $G$ and $H$.

A graph $G$ is \emph{$k$-WL-homogeneous} if $\WL{k}{G}(v,\dots,v) = \WL{k}{G}(w,\dots,w)$ for all $v,w\in V(G)$.

\subsection{Group Theory}
\label{sec:groups}

For a general background on group theory we refer to \cite{Rotman99}, whereas background on permutation groups can be found in \cite{DixonM96}.
Also, basics facts on algorithms for permutation groups are given in \cite{Seress03}.

\paragraph{Basics for Permutation Groups.}

A \emph{permutation group} acting on a set $\Omega$ is a subgroup $\Gamma \leq \Sym(\Omega)$ of the symmetric group.
The size of the permutation domain $\Omega$ is called the \emph{degree} of $\Gamma$.
If $\Omega = [n] \coloneqq \{1,\dots,n\}$, then we also write $S_n$ instead of $\Sym(\Omega)$.
For $A \subseteq \Omega$ and $\gamma \in \Gamma$ let $\gamma(A) \coloneqq \{\gamma(\alpha) \mid \alpha \in A\}$.
The set $A$ is \emph{$\Gamma$-invariant} if $\gamma(A) = A$ for all $\gamma \in \Gamma$.

Let $\theta\colon \Omega \rightarrow \Omega'$ be a bijection.
We write $\Gamma\theta \coloneqq \{\gamma\theta \mid \gamma \in \Gamma\}$ for the set of bijections from $\Omega$ to $\Omega'$ obtained from concatenating a permutation from $\Gamma$ and $\theta$.
Note that $(\gamma\theta)(\alpha) = \theta(\gamma(\alpha))$ for all $\alpha \in \Omega$.

A set $S \subseteq \Gamma$ is a \emph{generating set} for $\Gamma$ if for every $\gamma \in \Gamma$ there are $\delta_1,\dots,\delta_k \in S$ such that $\gamma = \delta_1 \dots \delta_k$.
In order to perform computational tasks for permutation groups efficiently the groups are represented by generating sets of small size (i.e., polynomial in the size of the permutation domain).
Indeed, most algorithms are based on so-called strong generating sets,
which can be chosen of size quadratic in the size of the permutation domain of the group and can be computed in polynomial time given an arbitrary generating set (see, e.g., \cite{Seress03}).

\paragraph{Group-Theoretic Methods for Isomorphism Testing.}

In this work, we shall be interested in a particular subclass of permutation groups.
Let $\Gamma$ be a group and let $\gamma,\delta \in \Gamma$.
The \emph{commutator} of $\gamma$ and $\delta$ is $[\gamma,\delta] \coloneqq \gamma^{-1}\delta^{-1}\gamma\delta$.
The \emph{commutator subgroup $[\Gamma,\Gamma]$} of $\Gamma$ is the unique subgroup of $\Gamma$ generated by all commutators $[\gamma,\delta]$ for $\gamma,\delta \in \Gamma$.
Note that $[\Gamma,\Gamma]$ is a normal subgroup of $\Gamma$.
The \emph{derived series of $\Gamma$} is the sequence of subgroups $\Gamma^{(0)} \trianglerighteq \Gamma^{(1)} \trianglerighteq \Gamma^{(2)} \trianglerighteq \dots$ where $\Gamma^{(0)} \coloneqq \Gamma$ and $\Gamma^{(i+1)} \coloneqq [\Gamma^{(i)},\Gamma^{(i)}]$ for all $i \geq 0$.
A group $\Gamma$ is \emph{solvable} if there is some $i \geq 0$ such that $\Gamma^{(i)}$ is the trivial group (i.e., it only contains the identity element).

By the Feit-Thompson Theorem every group of odd order is solvable.
Also, for every tournament $T$, the automorphism group $\Aut(T)$ has odd order.
Indeed, $\Aut(T)$ cannot contain an involution (a permutation of order $2$), since every involution swaps some pair $v,w$ of distinct vertices and exactly one of $(v,w), (w,v)$ is an edge of $T$.
Together, we obtain the following:

\begin{theorem}
 \label{thm:aut-solvable}
 Let $T$ be a tournament.
 Then $\Aut(T)$ is solvable.
\end{theorem}

Next, we state several basic group-theoretic algorithms for isomorphism testing.

\begin{theorem}[{\cite[Theorem 4.1]{BabaiL83}}]
 \label{thm:ti}
 There is an algorithm that, given two tournaments $T_1$ and $T_2$, computes $\Iso(T_1,T_2)$ in time $n^{O(\log n)}$.
\end{theorem}

Note that $\Iso(T_1,T_2)$ may be of size exponential in the number of vertices of $T_1$ and $T_2$.
However, if $T_1$ and $T_2$ are isomorphic (i.e., $\Iso(T_1,T_2) \neq \emptyset$), we have $\Iso(T_1,T_2) = \Aut(T_1) \varphi$ where $\varphi \in \Iso(T_1,T_2)$ is an arbitrary isomorphism from $T_1$ to $T_2$.
Hence, the set $\Iso(T_1,T_2)$ can be represented by a generating set for $\Aut(T_1)$ of size polynomial in $|V(T_1)|$ and a single element $\varphi \in \Iso(T_1,T_2)$.
Let us stress at this point that all isomorphism sets computed by the various algorithms discussed in this work are represented in this way.

Let $G_1$ and $G_2$ be two (colored) directed graphs.
Also let $\Gamma \leq \Sym(V(G_1))$ be a permutation group and let $\theta\colon V(G_1) \to V(G_2)$ be a bijection.
We define
\[\Iso_{\Gamma\theta}(G_1,G_2) \coloneqq \Iso(G_1,G_2) \cap \Gamma\theta = \{\varphi \in \Gamma\theta \mid \varphi\colon G_1 \cong G_2\}\]
and $\Aut_\Gamma(G_1) \coloneqq \Iso_\Gamma(G_1,G_1)$.
Note that $\Aut_\Gamma(G_1) \leq \Gamma$ and moreover, if $\Iso_{\Gamma\theta}(G_1,G_2) \neq \emptyset$, then $\Iso_{\Gamma\theta}(G_1,G_2) = \Aut_\Gamma(G_1)\varphi$ where $\varphi \in \Iso_{\Gamma\theta}(G_1,G_2)$ is an arbitrary isomorphism from $G_1$ to $G_2$.

\begin{theorem}[{\cite[Corollary 3.6]{BabaiL83}}]
 \label{thm:gi-solvable-group}
 Let $G_1 = (V_1,E_1,\lambda_1)$ and $G_1 = (V_2,E_2,\lambda_2)$ be two arc-colored directed graphs.
 Also let $\Gamma \leq \Sym(V_1)$ be a solvable group and $\theta\colon V_1 \to V_2$ a bijection.
 Then $\Iso_{\Gamma\theta}(G_1,G_2)$ can be computed in polynomial time.
\end{theorem}

A \emph{hypergraph} is a pair $\CH = (V,\CE)$ where $V$ is a finite, non-empty set of vertices and $\CE \subseteq 2^V$ is a subset of the powerset of $V$.
Two hypergraphs $\CH_1 = (V_1,\CE_1)$ and $\CH_2 = (V_2,\CE_2)$ are \emph{isomorphic} if there is a bijection $\varphi\colon V_1 \to V_2$ such that $E \in \CE_1$ if and only if $\varphi(E) \in \CE_2$ for all subsets $E \subseteq V$.
We write $\varphi\colon \CH_1 \cong \CH_2$ to denote that $\varphi$ is an \emph{isomorphism} from $\CH_1$ to $\CH_2$.
As before, we write $\Iso(\CH_1,\CH_2)$ to denote the set of all isomorphisms from $\CH_1$ to $\CH_2$.
Also, for a permutation group $\Gamma \leq \Sym(V_1)$ and a bijection $\theta\colon V(G_1) \to V(G_2)$,
we define
\[\Iso_{\Gamma\theta}(\CH_1,\CH_2) \coloneqq \Iso(\CH_1,\CH_2) \cap \Gamma\theta = \{\varphi \in \Gamma\theta \mid \varphi\colon \CH_1 \cong \CH_2\}\]
and $\Aut_\Gamma(\CH_1) = \Iso_\Gamma(\CH_1,\CH_1)$.

\begin{theorem}[\cite{Miller83}]
 \label{thm:hi-solvable-group}
 Let $\CH_1 = (V_1,\CE_1)$ and $\CH_2 = (V_2,\CE_2)$ be two hypergraphs.
 Also let $\Gamma \leq \Sym(V_1)$ be a solvable group and $\theta\colon V_1 \to V_2$ a bijection.
 Then $\Iso_{\Gamma\theta}(\CH_1,\CH_2)$ can be computed in polynomial time.
\end{theorem}

\paragraph{Wreath Products.}

Finally, we describe wreath products of groups which repeatedly appear within the main algorithm of this paper.
For simplicity, we focus on the application cases appearing in this work.

Let $G_1,\dots,G_\ell$ be directed graphs with pairwise disjoint vertex sets $V_1,\dots,V_\ell$.
Also assume that $G_i \cong G_j$ for all $i,j \in [\ell]$.
Then
\[\Aut(G_j) = \varphi^{-1}\Aut(G_i)\varphi \coloneqq \{\varphi^{-1}\gamma\varphi \mid \gamma \in \Aut(G_i)\}\]
for all $\varphi \in \Iso(G_i,G_j)$.
Now let $\Delta \leq S_\ell$ be a permutation group with domain $\{1,\dots,\ell\}$.
Also let $V \coloneqq V_1 \cup \dots \cup V_\ell$.
We define $\Gamma \leq \Sym(V)$ to be the permutation group containing all elements $\gamma \in \Sym(V)$ such that there are $\delta \in \Delta$ and, for each $i \in [\ell]$, an isomorphism $\varphi_i \in \Iso(G_i,G_{\delta(i)})$ such that for all $v \in V$ it holds that
\[\gamma(v) = \varphi_i(v)\]
where $i \in [\ell]$ is the unique index such that $v \in V_i$.
The group $\Gamma$ is the \emph{wreath product} of $\Aut(G_1)$ by $\Delta$.
Observe that $\Gamma$ is solvable if and only if $\Delta$ and $\Aut(G_1)$ are solvable.

Now let $S_1,\dots,S_\ell$ be generating sets for $\Aut(G_1),\dots,\Aut(G_\ell)$, respectively.
Also let $S_\Delta$ be a generating set for $\Delta$.
Then we obtain a generating set $S_\Gamma$ as follows.
For every $\gamma_i \in S_i$ we define $\gamma_i^*$ to be the extension of $\gamma_i$ to $V$ which fixes all elements outside of $V_i$.
More formally, $\gamma_i^*(v) = \gamma_i(v)$ for every $v \in V_i$, and $\gamma_i^*(v) = v$ for every $v \in V \setminus V_i$.
Let $S_i^* \coloneqq \{\gamma_i^* \mid \gamma_i \in S_i\}$.

Additionally, for every $\delta \in S_\Delta$, we fix an element $\varphi_i \in \Iso(G_i,G_{\delta(i)})$ for every $i \in [\ell]$.
We set $\delta^*(v) \coloneqq \varphi_i(v)$ for every $v \in V$ where $i \in [\ell]$ is the unique index such that $v \in V_i$.
Let $S_\Delta^* \coloneqq \{\delta^* \mid \delta \in S_\Delta\}$.
Then
\[S_\Gamma \coloneqq S_1^* \cup \dots \cup S_\ell^* \cup S_\Delta^*\]
is a generating of $\Gamma$.
Note that $S_\Gamma$ can be computed in polynomial time given the generating sets $S_1,\dots,S_\ell$ and $S_\Delta$.

\section{Small Degree Partition Sequences}
\label{sec:partitions}

In the following, we design an isomorphism test for tournaments of twin width $k$ which runs in time $k^{O(\log k)}n^{O(1)}$.
On a high level, the algorithm essentially proceeds in three phases.
First, we use well-established group-theoretic methods going back to \cite{BabaiL83,Luks82} to reduce to the case where both input tournaments are $2$-WL-homogeneous (without increasing the twin width).
In the second step, we identify a substructure of an input tournament $T$ (that is $2$-WL-homogeneous) that has some kind of bounded-degree property.
More concretely, we apply the $2$-dimensional Weisfeiler-Leman algorithm and compute a sequence of colors $c_1,\dots,c_\ell$ in the image of the $2$-WL coloring $\WL{2}{T}$ so that the subgraph induced by all edges with a color from $c_1,\dots,c_\ell$ has a certain type of bounded-degree property.
After that, we rely on the computed bounded-degree structure to determine isomorphisms based on the group-theoretic graph isomorphism machinery.
Similar tools have also been used in \cite{ArvindPR25} to solve isomorphism of $k$-spanning tournaments.
However, as we shall see below, the bounded-degree property guaranteed by the second step is weaker than the notion of $k$-spanning tournaments, which requires us to further extend the methods from~\cite{ArvindPR25}.

In this section, we implement the second phase and prove the key combinatorial lemma (Lemma~\ref{lem:partition-sequence}) underlying our isomorphism algorithm.
Our arguments rely on the notion of \emph{mixed neighbors} for a pair of vertices.
For a pair $v,w\in V(T)$ of vertices we let
\begin{equation}
 M(v,w) \coloneqq \Big(N_-(v)\cap N_+(w)\Big)\cup \Big(N_+(v)\cap N_-(w)\Big).
\end{equation}
We call the elements of $M(v,w)$ the \emph{mixed neighbors} of $(v,w)$,
and we call $\md(v,w) \coloneqq |M(v,w)|$ the \emph{mixed degree} of $(v,w)$.
The following simple observation links the mixed degree to twin width.

\begin{observation}
 \label{obs:edge-small-mixed-degree}
 There is an edge $(v,w) \in E(T)$ such that $\md(v,w) \leq \tww(T)$.
\end{observation}

\begin{proof}
 Let $k \coloneqq \tww(T)$ and let $\CP_1,\dots,\CP_n$ be a contraction sequence of $T$ of width $k$.
 Let $\{v,w\}$ be the unique $2$-element part in $\CP_2$.
 Then $\md(v,w)=\md(w,v) \leq k$, and either $(v,w)\in E(T)$ or $(w,v)\in E(T)$.
\end{proof}

In the following, let $G_T$ be the directed graph with vertex set $V(G_T) \coloneqq V(T)$ and edge set $E(G_T) \coloneqq \{(v,w) \in E(T) \mid \md(v,w) \leq \tww(T)\}$.
The next lemma implies that $G_T$ has maximum out-degree at most $2\cdot\tww(T) + 1$.

\begin{lemma}
 \label{lem:mixed-degree-few-neighbors}
 Suppose $k \geq 1$.
 Let $T$ be a tournament and let $v \in V(T)$.
 Also let
 \[W \coloneqq \{w \in N_+(v) \mid \md(v,w) \leq k\}.\]
 Then $|W| \leq 2k + 1$.
\end{lemma}

\begin{proof}
 Let $\ell \coloneqq |W|$. The induced subtournament $T[W]$ has
 a vertex $w$ of in-degree at least $(\ell-1)/2$.
 Since $\md(v,w) \leq k$ and $(v,w') \in E(T)$ for all $w' \in W$, we have
 \[|\{w' \in W \mid (w',w) \in E(T)\}| \leq k.\]
 Thus $\frac{\ell - 1}{2} \leq k$, which implies that $|W| = \ell \leq 2k + 1$.
\end{proof}

So $G_T$ is a subgraph of $T$ of maximum out-degree $d \coloneqq 2\cdot\tww(T) + 1$.
We remark that similar arguments also show that $G_T$ has maximum in-degree at most $d$ (technically, this property is not required by our algorithm, but it is helpful for the following explanations).
Now, first suppose that $G_T$ is strongly connected.
Then the edges of $G_T$ define a (strongly) connected spanning subgraph of maximum degree $2d$ (in-degree plus out-degree).
In this situation, we can directly use the algorithm from \cite{ArvindPR25} to test isomorphism in time $d^{O(\log d)}n^{O(1)}$.

So suppose $G_T$ is not strongly connected.
If $T$ is $2$-WL-homogeneous, then $G_T$ is also not weakly connected (i.e., the undirected version of $G_T$ is not connected); see Lemma \ref{lem:weakly-to-strongly-cc}.
In this case, the basic idea is to identify further edges to be added to decrease the number of connected components while keeping some kind of bounded-degree property.

In the following, let $\CQ$ be a partition of $V(T)$ that is non-trivial, that is, has at least two parts.
The reader is encouraged to think of $\CQ$ as the partition into the (weakly) connected components of $G_T$, but the following results hold for any non-trivial partition $\CQ$.
We call an edge $(v,v') \in E(T)$ \emph{cross-cluster with respect to $\CQ$} if it connects distinct $Q,Q' \in \CQ$.
For a cross-cluster edge $(v,v')$ with $Q \owns v, Q'\owns v'$, we let
\[\CM_{\CQ}(v,v') \coloneqq \big\{ Q'' \in \CQ \setminus \{Q,Q'\} \bigmid Q'' \cap M(v,v') \neq \emptyset \big\}\]
and $\md_\CQ(v,v') \coloneqq |\CM_\CQ(v,v')|$.

The next two lemmas generalize Observation \ref{obs:edge-small-mixed-degree} and Lemma \ref{lem:mixed-degree-few-neighbors}.

\begin{lemma}
 \label{lem:edge-small-mixed-degree}
 Let $T$ be a tournament and suppose $\CQ$ is a non-trivial partition of $V(T)$.
 Then there is a cross-cluster edge $(v,w)\in E(T)$ such that $\md_\CQ(v,w) \leq \tww(T)$.
\end{lemma}

\begin{proof}
 Let $k \coloneqq \tww(T)$ and let $\CP_1,\dots,\CP_n$ be a contraction sequence of $T$ of width $k$.
 Note that $\CP_1$ refines $\CQ$ and $\CP_n$ does not refine $\CQ$, because $\CQ$ is nontrivial.
 Let $i \geq 1$ be minimal such that $\CP_{i+1}$ does not refine $\CQ$.

 Let $P,P' \in\CP_i$ denote the parts merged in the step from $\CP_i$ to $\CP_{i+1}$.
 Since $\CP_i$ refines $\CQ$, there are $Q,Q' \in \CQ$ such that $P \subseteq Q$ and $P'\subseteq Q'$.
 Moreover, $Q \neq Q'$, because $\CP_{i+1}$ does not refine $\CQ$.
 We pick arbitrary elements $v \in P$ and $w \in P'$ such that $(v,w) \in E(T)$ (if $(w,v)\in E(T)$, we swap the roles of $P$ and $P'$).
 Then $(v,w)$ is a cross-cluster edge with respect to $\CQ$.

 Let $P_1,\dots,P_{k'}$ be a list of all $P'' \in \CP_{i+1} \setminus \{P\cup P'\}$ such that the pair $(P\cup P',P'')$ is not homogeneous.
 Then $k' \leq k$ by the definition of twin width.
 Since $\CP_i \setminus \{P,P'\} = \CP_{i+1} \setminus \{P\cup P'\}$ and $\CP_i$ refines $\CQ$, there are $Q_1,\dots,Q_{k'} \in \CQ$ such that $P_i \subseteq Q_i$ for all $i \in [k']$.

 Now let $Q'' \in \CQ \setminus \{Q,Q',Q_1,\dots,Q_{k'}\}$.
 Suppose for contradiction that $Q'' \cap M(v,w) \neq \emptyset$, and pick an element $w' \in Q'' \cap M(v,w)$.
 Then there is a $P'' \in \CP_i \setminus \{P,P',P_1,\dots,P_{k'}\} = \CP_{i+1} \setminus \{P\cup P',P_1,\dots,P_{k'}\}$ such that $w' \in P''$.
 But then the pair $(P\cup P',P'')$ is not homogeneous, which is a contradiction.
 So $Q'' \cap M(v,w) = \emptyset$.
 This implies that $\CM_\CQ(v,w) \subseteq \{Q_1,\dots,Q_{k'}\}$.
 In particular, $\md_\CQ(v,w) \leq k' \leq k$.
\end{proof}

\begin{lemma}
 \label{lem:mixed-degree-few-neighbors-partition}
 Suppose $k \geq 1$.
 Let $T$ be a tournament and let $\CQ$ be a non-trivial partition of $V(T)$.
 Also let $Q \in \CQ$ and $v \in Q$.
 Let
 \[\CW \coloneqq \{Q' \in \CQ \setminus \{Q\} \mid \exists w \in Q'\colon (v,w) \in E(T) \wedge \md_\CQ(v,w) \leq k\}.\]
 Then $|\CW| \leq 2k + 1$.
\end{lemma}

The proof is very similar to the proof of Lemma~\ref{lem:mixed-degree-few-neighbors}.

\begin{proof}
 Let $\ell \coloneqq |\CW|$ and suppose $\CW=\{Q_1,\dots,Q_\ell\}$.
 For every $i \in [\ell]$ pick an element $w_i \in Q_i$ such that $(v,w_i) \in E(T)$ and $\md_\CQ(v,w_i) \leq k$.
 We define $W \coloneqq \{w_1,\dots,w_\ell\}$.
 Then there is some $w \in W$ such that
 \[|\{w' \in W \mid (w',w) \in E(T)\}| \geq \frac{\ell - 1}{2},\]
 because the induced subtournament $T[W]$ has a vertex of in-degree at least $(\ell-1)/2$.

 Since $\md_\CQ(v,w) \leq k$ and $(v,w') \in E(T)$ for all $w' \in W$, it follows that
 \[|\{w' \in W \mid (w',w) \in E(T)\}| \leq k.\]
 Thus $\frac{\ell - 1}{2} \leq k$, which implies that $|\CW| = \ell \leq 2k + 1$.
\end{proof}

\begin{figure}
 \centering
 \begin{tikzpicture}[scale=0.9]
  \draw[gray!60, fill=gray!60] (0,0) circle (1.4cm);
  \foreach \j in {0,...,5}{
   \node[vertex] (v6-\j) at (60*\j:1.2) {};
  }
  \foreach \v/\w in {0/1,1/2,2/3,3/4,4/5,5/0}{
   \draw[->, line width = 1.5pt, blue] (v6-\v) to (v6-\w);
  }

  \foreach \i in {0,...,5}{
   \draw[gray!60, fill=gray!60] (60*\i:3.2) circle (1.4cm);
   \foreach \j in {0,...,5}{
    \node[vertex] (v\i-\j) at ($(60*\i:3.2) + (30 + 60*\j:1.2)$) {};
   }
   \foreach \v/\w in {0/1,1/2,2/3,3/4,4/5,5/0}{
    \draw[->, line width = 1.5pt, blue] (v\i-\v) to (v\i-\w);
   }
  }

  \draw[->, line width = 1.5pt, Green] (v6-0) to (v0-0);
  \draw[->, line width = 1.5pt, Green, bend left] (v6-0) to (v0-1);
  \draw[->, line width = 1.5pt, Green] (v6-0) to (v0-2);
  \draw[->, line width = 1.5pt, Green] (v6-0) to (v0-3);
  \draw[->, line width = 1.5pt, Green, bend right] (v6-0) to (v0-4);
  \draw[->, line width = 1.5pt, Green] (v6-0) to (v0-5);

  \draw[->, line width = 1.5pt, Green] (v6-1) to (v1-0);
  \draw[->, line width = 1.5pt, Green] (v6-1) to (v1-1);
  \draw[->, line width = 1.5pt, Green, bend left] (v6-1) to (v1-2);
  \draw[->, line width = 1.5pt, Green] (v6-1) to (v1-3);
  \draw[->, line width = 1.5pt, Green] (v6-1) to (v1-4);
  \draw[->, line width = 1.5pt, Green, bend right] (v6-1) to (v1-5);

  \draw[->, line width = 1.5pt, Green, bend right] (v6-2) to (v2-0);
  \draw[->, line width = 1.5pt, Green] (v6-2) to (v2-1);
  \draw[->, line width = 1.5pt, Green] (v6-2) to (v2-2);
  \draw[->, line width = 1.5pt, Green, bend left] (v6-2) to (v2-3);
  \draw[->, line width = 1.5pt, Green] (v6-2) to (v2-4);
  \draw[->, line width = 1.5pt, Green] (v6-2) to (v2-5);

  \draw[->, line width = 1.5pt, Green] (v6-3) to (v3-0);
  \draw[->, line width = 1.5pt, Green, bend right] (v6-3) to (v3-1);
  \draw[->, line width = 1.5pt, Green] (v6-3) to (v3-2);
  \draw[->, line width = 1.5pt, Green] (v6-3) to (v3-3);
  \draw[->, line width = 1.5pt, Green, bend left] (v6-3) to (v3-4);
  \draw[->, line width = 1.5pt, Green] (v6-3) to (v3-5);

  \draw[->, line width = 1.5pt, Green] (v6-4) to (v4-0);
  \draw[->, line width = 1.5pt, Green] (v6-4) to (v4-1);
  \draw[->, line width = 1.5pt, Green, bend right] (v6-4) to (v4-2);
  \draw[->, line width = 1.5pt, Green] (v6-4) to (v4-3);
  \draw[->, line width = 1.5pt, Green] (v6-4) to (v4-4);
  \draw[->, line width = 1.5pt, Green, bend left] (v6-4) to (v4-5);

  \draw[->, line width = 1.5pt, Green, bend left] (v6-5) to (v5-0);
  \draw[->, line width = 1.5pt, Green] (v6-5) to (v5-1);
  \draw[->, line width = 1.5pt, Green] (v6-5) to (v5-2);
  \draw[->, line width = 1.5pt, Green, bend right] (v6-5) to (v5-3);
  \draw[->, line width = 1.5pt, Green] (v6-5) to (v5-4);
  \draw[->, line width = 1.5pt, Green] (v6-5) to (v5-5);

  \foreach \j in {0,...,5}{
   \node at (30 + 60*\j:4.4) {$\dots$};
  }
 \end{tikzpicture}
 \caption{The figure shows part of a tournament $T$. The colors $c_1$ and $c_2$ are shown in blue and green, respectively.
  Also, the parts of the partition $\CQ_1$ are highlighted in gray.
  Note that only green edges, which are outgoing from the middle part, are shown.}
 \label{fig:partition-sequence}
\end{figure}
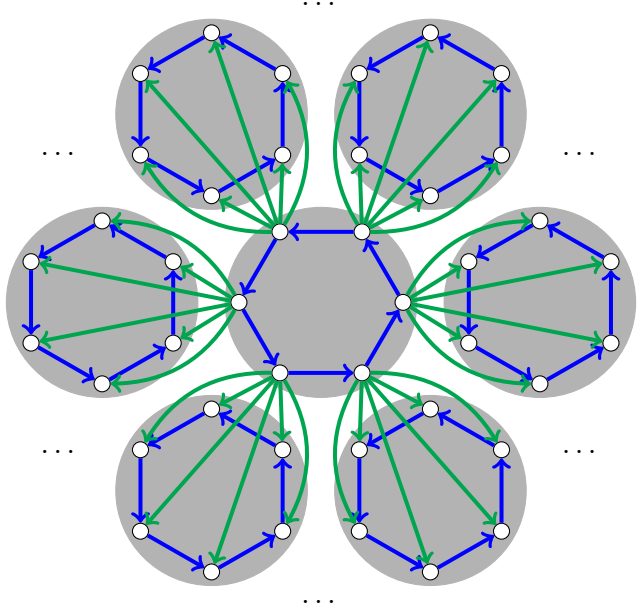

Now, suppose we color all edges $(v,w)$ of $T$ with $\md(v,w) \leq \tww(T)$ using the color $c_1 = \blue$ (see Figure \ref{fig:partition-sequence}).
Let $\CQ_1$ be the partition into the (weakly) connected components of the graph induced by the blue edges and suppose that $\CQ_1$ is non-trivial.
We can compute isomorphisms between the different parts of $\CQ_1$ using the algorithm from \cite{ArvindPR25}.
Next, let us color all cross-cluster edges $(v,w)\in E(T)$ with $\md_\CQ(v,w) \leq \tww(T)$ using the color $c_2 = \green$.
Then every vertex has outgoing green edges to at most $2\tww(T) + 1$ other parts of $\CQ_1$ (see Lemma \ref{lem:mixed-degree-few-neighbors-partition}).
However, since a vertex may have an unbounded number of green neighbors in a single part, the out-degree of the graph induced by the green edges may be unbounded.
So it is not possible to use the algorithm from \cite{ArvindPR25} as a black-box on the components induced by blue and green edges.
Luckily, the methods used in \cite{ArvindPR25} can be extended to work even in this more general setting (see Section \ref{sec:alg}).
So if the graph induced by the blue and green edges is connected, then we are again done.
Otherwise, we let $\CQ_2$ denote the partition into (weakly) connected components of the graph induced by the blue and green edges.
Now, we can continue in the same fashion identifying colors $c_3,c_4,\dots$ and corresponding partitions $\CQ_3,\CQ_4,\dots$ until the graph induced by all edges of colors $c_1,\dots,c_\ell$ is eventually connected.

Below, we provide a lemma that computes the corresponding sequence of partitions and edge colors using $2$-WL.
To state the lemma in its cleanest form, we restrict our attention to tournaments that are $2$-WL-homogeneous.
Recall that a tournament $T$ is $2$-WL-homogeneous if for all $v,w\in V(T)$ is holds that $\WL{2}{T}(v,v) = \WL{2}{T}(w,w)$.

We also require another piece of notation.
For a directed graph $G$ and a set of colors $C \subseteq \{\WL{2}{G}(v,w) \mid (v,w) \in E(G)\}$ we write $G[C]$ for the directed graph with vertex set $V(G[C]) \coloneqq V(G)$ and edge set
\[E(G[C]) \coloneqq \{(v,w) \in E(G) \mid \WL{2}{G} \in C\}.\]
To prove Lemma~\ref{lem:partition-sequence} we need the following lemma about the connected components of the graphs $G[C]$.

\begin{lemma}
 \label{lem:weakly-to-strongly-cc}
 Let $G$ be a $2$-WL-homogeneous graph, and let $C$ be a set of colors in the range of $\WL{2}{G}$.
 Then the weakly connected components of $G[C]$ equal the strongly connected components of~$G[C]$.
\end{lemma}

\begin{proof}
 Let $A$ be a weakly connected component of $G[C]$ and suppose towards a contradiction that $A$ is not strongly connected in $G[C]$.
 Let $B_1,\dots,B_\ell$ be the strongly connected components of $G[C]$ within $A$.
 We may assume that $B_1,\dots,B_\ell$ are ordered topologically, i.e., there is an edge in $G[C]$ from a vertex in $B_i$ to a vertex in $B_j$ only if $i < j$.
 Now consider $B \coloneqq B_1$, i.e., $B \subseteq A$ is a strongly connected component of $G[C]$ such that no edges of a color from $C$ are incoming into $B$.
 Also, let $w \in A \setminus B$ denote a vertex such that there is a directed path from $v$ to $w$ for every $v \in B$.
 Also, fix a vertex $v \in B$.

 Now, let $C^{\rightarrow}$ denote the set of all colors $\WL{2}{G}(x,y)$ where $x,y \in V(G)$ are distinct vertices such that there is a path from $x$ to $y$ in $G[C]$, but no path from $y$ to $x$.
 By the properties of $2$-WL, for every $x,y \in V(G)$ such that $\WL{2}{G}(x,y) \in C^{\rightarrow}$, there is a path from $x$ to $y$ in $G[C]$, but no path from $y$ to $x$.

 By definition, we have that $\WL{2}{G}(v,w) \in C^{\rightarrow}$, but there is no $u \in V(G)$ such that $\WL{2}{G}(u,v) \in C^{\rightarrow}$ (no edges of a color from $C$ are incoming into $B$).
 It follows that $\WL{2}{G}(v,v) \neq \WL{2}{G}(w,w)$ which is a contradiction.
\end{proof}

\begin{lemma}
 \label{lem:partition-sequence}
 Let $T$ be a $2$-WL-homogeneous tournament of twin width $\tww(T) \leq k$.
 Then there is a sequence of partitions $\{\{v\} \mid v \in V(T)\} = \CQ_0,\dots,\CQ_\ell = \{V(T)\}$ of $V(T)$ where $\CQ_{i-1}$ refines $\CQ_i$ for all $i$,
 and a sequence of colors $c_1,\dots,c_\ell$ in the range of $\WL{2}{T}$ such that
 \begin{enumerate}
  \item\label{item:partition-sequence-1} $\CQ_i$ is the partition into the strongly connected components of $T[\{c_1,\ldots,c_i\}]$ for every $i \in [\ell]$, and
  \item\label{item:partition-sequence-2} for every $i \in [\ell]$ and every $v \in V(T)$ it holds that
  \[\big|\big\{Q \in \CQ_{i-1} \bigmid \exists w \in Q \colon \WL{2}{T}(v,w) = c_i\big\}\big| \leq 2k+1.\]
 \end{enumerate}
 Moreover, there is a polynomial-time algorithm that, given a tournament $T$ and an integer $k \geq 1$, computes the desired sequences $\CQ_0,\dots,\CQ_\ell$ and $c_1,\dots,c_\ell$ or concludes that $\tww(T) > k$.
\end{lemma}

\begin{proof}
 We set $\CQ_0 = \{\{v\} \mid v \in V(T)\}$ and inductively define a sequence of partitions and colors as follows.
 Let $i \geq 0$ and suppose we already defined partitions $\CQ_0 \prec \dots \prec \CQ_i$ and colors $c_1,\dots,c_i$.
 If $\CQ_i = \{V(T)\}$, we set $\ell \coloneqq i$ and complete both sequences.
 Otherwise, there is a cross-cluster edge $(v_{i+1},w_{i+1})$ with respect to $\CQ_i$ such that $\md_{\CQ_i}(v_{i+1},w_{i+1}) \leq k$ by Lemma~\ref{lem:edge-small-mixed-degree}.
 We set $c_{i+1} \coloneqq \WL{2}{T}(v_{i+1},w_{i+1})$ and define $\CQ_{i+1}$ to be the set of weakly connected components of $T[c_1,\dots,c_{i+1}]$.
 By Lemma~\ref{lem:weakly-to-strongly-cc}, these are also the strongly connected components.

 First observe that $\CQ_i \prec \CQ_{i+1}$ since $(v_{i+1},w_{i+1})$ is a cross-cluster edge with respect to $\CQ_i$, and $v_{i+1},w_{i+1}$ are contained in the same part of $\CQ_{i+1}$ by Lemma \ref{lem:weakly-to-strongly-cc}.
 Also, Property \ref{item:partition-sequence-1} is satisfied by definition.
 For Property \ref{item:partition-sequence-2} note that every edge $(v,w) \in E(T)$ such that $\WL{2}{G}(v,w) = c_{i+1}$ is a cross-cluster edge with respect to $\CQ_i$.
 So Property \ref{item:partition-sequence-2} follows directly from Lemma \ref{lem:mixed-degree-few-neighbors-partition}.

 Finally, it is clear from the description above that $\CQ_0 \prec \dots \prec \CQ_\ell$ and $c_1,\dots,c_\ell$ can be computed in polynomial time, or we conclude that $\tww(T) > k$.
\end{proof}

\section{The Isomorphism Algorithm}
\label{sec:alg}

Based on the structural insights summarized in Lemma \ref{lem:partition-sequence}, we now design an isomorphism test for tournaments of small twin width.

The strategy of our algorithm is the following.
We are given two tournaments $T_1$ and $T_2$, and we want to compute $\Iso(T_1,T_2)$.
First, we reduce to the case where both $T_1$ and $T_2$ are $2$-WL-homogeneous.

Towards this end, we start by applying $2$-WL and, for $j=1,2$, compute the coloring $\WL{2}{T_j}$.
If $2$-WL distinguishes the two tournaments, we can immediately conclude that they are non-isomorphic and return $\Iso(T_1,T_2) = \emptyset$.

So suppose that $2$-WL does not distinguish the tournaments.
Then $T_1$ is $2$-WL-homogeneous if and only if $T_2$ is $2$-WL-homogeneous.

If $T_1$ and $T_2$ are not $2$-WL-homogeneous, we rely on the following standard argument.
Let $c_1,\ldots,c_p$ be the vertex colors.
For $i \in [p]$ and $j=1,2$, let $P_{j,i}$ be the set of all $v\in V(T_j)$ such that $\WL{2}{T_j}(v,v) = c_i$.
We recursively compute the sets $\Lambda_i \coloneqq \Iso(T_1[P_{1,i}],T_2[P_{2,i}])$ for all $i \in [p]$.
If there is some $i \in [p]$ such that $\Lambda_i = \emptyset$, then $T_1$ and $T_2$ are non-isomorphic, and we return $\Iso(T_1,T_2)=\emptyset$.
Otherwise, the set $\Lambda_i$ is a coset of $\Gamma_i \coloneqq \Aut(T_1[P_{1,i}])$ for all $i \in [p]$, i.e., $\Lambda_i = \Gamma_i\theta_i$ for some bijection $\theta_i\colon P_{1,i} \to P_{2,i}$.
As the automorphism group of a tournament, $\Gamma_i$ is solvable (see Theorem \ref{thm:aut-solvable}).
Moreover, since the color classes $P_{1,i}$ are invariant under automorphisms of $T_1$, the automorphism group $\Gamma \coloneqq \Aut(T_1)$ is a subgroup of the direct product $\prod_i \Gamma_{i}$, which is also a solvable group.
Also, $\Iso(T_1,T_2) \subseteq \Gamma\theta$ where $\theta\colon V(T_1) \to V(T_2)$ is the unique bijection defined via $\theta(v) \coloneqq \theta_i(v)$ for all $v \in V(T_1)$, where $i \in [p]$ is the unique index such that $v \in P_{1,i}$.
So $\Iso(T_1,T_2) = \Iso_{\Gamma\theta}(T_1,T_2)$ can be computed in polynomial time using Theorem \ref{thm:gi-solvable-group}.

So we may assume that $T_1$ and $T_2$ are $2$-WL-homogeneous.
In this case, we apply Lemma~\ref{lem:partition-sequence} and obtain colors $c_1,\ldots,c_\ell$ and, for $j=1,2$, a partition sequence $\{\{v\} \mid v \in V(T_j)\} = \CQ_{j,0}, \dots, \CQ_{j,\ell} = \{V(T_j)\}$ where $\CQ_{j,i-1}$ refines $\CQ_{j,i}$ for all $i \in [\ell]$.

Now, we iteratively compute for $i=0,\dots,\ell$ the sets $\Iso(T_j[Q],T_{j'}[Q'])$ for all $j,j' \in \{1,2\}$ and all $Q \in \CQ_{j,i}$ and $Q' \in \CQ_{j',i}$.
For $i = 0$ this is trivial since all parts have size $1$.
The next lemma describes the key subroutine of the main algorithm which allows us to compute the isomorphism sets for level $i \in [\ell]$ given all the sets for level $i-1$.
Note that, on the last level $\ell$, we compute the set $\Iso(T_1,T_2)$ since $\CQ_{j,\ell} = \{V(T_j)\}$ for both $j \in \{1,2\}$.

To state the lemma, we need additional terminology.
Let $T = (V,E,\lambda)$ be an arc-colored tournament.
A partition $\CQ$ of $V$ is \emph{$\lambda$-definable} if there is a set of colors $C \subseteq \{\lambda(v,w) \mid v= w \vee (v,w) \in E\}$ such that
\[v \sim_\CQ w \quad\iff\quad \lambda(v,w) \in C\]
for all $v,w \in V$ such that $v = w$ or $(v,w) \in E$.
We also say that \emph{$\CQ$ is $\lambda$-defined by $C$}.
If $\CQ$ is $\lambda$-defined by $C$ then we can partition the colors in the range of $\lambda$ into the colors in $C$, which we call \emph{intra-cluster} colors, and the remaining colors, which we call \emph{cross-cluster} colors.
Note that if a color $c$ is intra-cluster, then for all $(v,w) \in E$ with $\lambda(v,w) = c$ it holds that $v,w \in Q$ for some $Q \in \CQ$,
and if $c$ is cross-cluster, then for all $(v,w) \in E$ with $\lambda(v,w) = c$ it holds that $(v,w)$ is a cross-cluster edge, that is, $v \in Q$ and $w \in Q'$ for distinct $Q,Q' \in \CQ$.

\begin{lemma}
 \label{lem:lift-isomorphisms}
 There is an algorithm that, given
 \begin{enumerate}[label = (\Alph*)]
  \item\label{item:lift-isomorphisms-1} an integer $d \geq 1$;
  \item\label{item:lift-isomorphisms-2} two arc-colored tournaments $T_1=(V_1,E_1,\lambda_1)$ and $T_2=(V_2,E_2,\lambda_2)$;
  \item\label{item:lift-isomorphisms-3} a set of colors $C$ and for $j=1,2$ a partition $\CQ_j$ of $V_j$ that is $\lambda_j$-defined by $C$;
  \item\label{item:lift-isomorphisms-4} a color $c^*$ that is cross-cluster with respect to $\CQ_j$ for $j=1,2$ and
   \begin{itemize}
    \item for every $v \in V_j$ it holds that
     \[\big|\big\{Q \in \CQ_j \bigmid \exists w \in Q\colon (v,w) \in E_j \wedge \lambda_j(v,w) = c^*\big\}\big| \leq d;\]
    \item for
     \[F_j \coloneqq \big\{(Q,Q')\in\CQ_j^2 \bigmid Q\neq Q', \exists w \in Q,w' \in Q'\colon (w,w') \in E_j \wedge \lambda_j(w,w') = c^*\big\}\]
     the directed graph $G_j = (\CQ_j,F_j)$ is strongly connected;
   \end{itemize}
  \item\label{item:lift-isomorphisms-5} $\Iso\big(T_j[Q],T_{j'}[Q']\big)$ for every $j,j' \in \{1,2\}$ and every $Q \in \CQ_j$, $Q' \in \CQ_{j'}$,
 \end{enumerate}
 computes $\Iso(T_1,T_2)$ in time $d^{O(\log d)} \cdot n^{O(1)}$.
\end{lemma}

\begin{proof}
 The algorithm fixes an arbitrary vertex $r_1 \in V_1$.
 For every $r_2 \in V_2$ we compute the set $\Iso((T_1,r_1),(T_2,r_2))$ of all isomorphisms $\varphi \in \Iso(T_1,T_2)$ such that $\varphi(r_1) = r_2$.
 Observe that
 \[\Iso(T_1,T_2) = \bigcup_{r_2 \in V_2} \Iso((T_1,r_1),(T_2,r_2)).\]
 So for the remainder of the proof, let us also fix some $r_2 \in V_2$.
 Let $R_1 \in \CQ_1$ be the unique set such that $r_1 \in R_1$, and similarly let $R_2 \in \CQ_2$ be the unique set such that $r_2 \in R_2$.

 We compute sets $\Iso((T_1[W_{1,i}],r_1),(T_2[W_{2,i}],r_2))$ for increasingly larger subsets $R_1 \subseteq W_{1,i} \subseteq V_1$ and $R_2 \subseteq W_{2,i} \subseteq V_2$.
 In other words, in each iteration $i \geq 0$, we build increasingly larger ``windows'' $W_{i,1}$ and $W_{i,2}$ and compute all isomorphisms between the sub-tournaments induced by $W_{i,1}$ and $W_{i,2}$.
 To obtain $W_{i+1,1}$ and $W_{i+1,2}$ in iteration $i+1$, we add a positive number of vertices to $W_{i,1}$ and $W_{i,2}$ and update the isomorphism set computed in the previous round to take the newly added vertices into account.
 Throughout, we maintain the property that
 \begin{enumerate}[label = (I.\arabic*),leftmargin=*]
  \item\label{item:invariant-1} there is some $\{R_j\} \subseteq \CW_{j,i} \subseteq \CQ_j$ such that $W_{j,i} = \bigcup \CW_{j,i}$, and
  \item\label{item:invariant-2} $\varphi(W_{1,i}) = W_{2,i}$ for every $\varphi \in \Iso((T_1,r_1),(T_2,r_2))$
 \end{enumerate}
 for every $i \geq 0$.

 Eventually, the algorithm either concludes that $\Iso((T_1,r_1),(T_2,r_2)) = \emptyset$ and terminates, or it reaches a point where $W_{i,1} = V_1$ and $W_{i,2} = V_2$ which gives the desired isomorphism set $\Iso((T_1,r_1),(T_2,r_2))$.

 For the base case $i = 0$ we initialize $\CW_{j,0} \coloneqq \{R_j\}$ for both $j \in \{1,2\}$.
 Also, we set $W_{j,0} \coloneqq R_j$.
 This means we are given the set $\Gamma_0'\theta_0' \coloneqq \Iso(T_1[W_{1,0}],T_2[W_{2,0}])$ as part of the input.
 Observe that $\Gamma_0'$ is solvable by Theorem \ref{thm:aut-solvable}.
 So
 \[\Gamma_0\theta_0 \coloneqq \Iso((T_1[W_{1,0}],r_1),(T_2[W_{2,0}],r_2))\]
 can be computed in polynomial time using Theorem \ref{thm:gi-solvable-group}.

 So assume $i \geq 0$ and let $W_{1,i} \subsetneq V_1$ and $W_{2,i} \subsetneq V_2$ be two sets satisfying \ref{item:invariant-1} and \ref{item:invariant-2}.
 Also suppose we have already computed the set
 \begin{equation}
  \label{eq:def-gamma-i}
  \Gamma_i\theta_i \coloneqq \Iso((T_1[W_{1,i}],r_1),(T_2[W_{2,i}],r_2)).
 \end{equation}
 We may assume the isomorphism set is non-empty; otherwise $\Iso((T_1,r_1),(T_2,r_2)) = \emptyset$ by Invariant \ref{item:invariant-2} and we terminate.
 For $j \in \{1,2\}$ let
 \[U_{j,i} \coloneqq \{u \in W_{j,i} \mid \exists w \in V_j \setminus W_{j,i} \colon (u,w) \in E_j \wedge \lambda_j(u,w) = c^*\}.\]
 Note that $U_{j,i} \neq \emptyset$ for both $j \in \{1,2\}$ by Property \ref{item:lift-isomorphisms-4}.
 Also, for every $j \in \{1,2\}$ and $u \in U_{j,i}$, we define
 \[\CL_{j,i+1}^{u} \coloneqq \{Q \in \CQ_j \setminus \CW_{j,i} \mid \exists w \in Q \colon (u,w) \in E_j \wedge \lambda_j(u,w) = c^*\}.\]
 We have
 \[|\CL_{j,i+1}^{u}| \leq d\]
 by Property \ref{item:lift-isomorphisms-4}.
 We define
 \[L_{j,i+1}^{u} \coloneqq \bigcup \CL_{j,i+1}^{u}\]
 and the tournament
 \[T_j^u \coloneqq T_j[L_{j,i+1}^{u}]\]
 for every $j \in \{1,2\}$ and $u \in U_{j,i}$.
 A visualization is also given in Figure \ref{fig:extend-isomorphisms}.

 \begin{figure}
  \centering
  \scalebox{1.1}{
  \begin{tikzpicture}
   \draw[fill = orange!20, orange!20] ($(6,5.5)!0.15!(-0.5,2.5)$) {[rounded corners] -- (-0.5,2.5)} -- ($(6,-1)!0.15!(-0.5,2.5)$) -- cycle;

   \draw[fill = blue!40, blue!40] (4.3,2.5) ellipse (0.5cm and 2cm);

   \draw[thick, rounded corners] (6,5.5) -- (-0.5,2.5) -- (6,-1);

   \node[smallvertex,label={left:{\small $r_1$}}] at (0.5,2.5) {};

   \node[smallvertex,label={left:{\small $u$}}] (u1) at (4.5,3.5) {};
   \node[smallvertex,label={left:{\small $u'$}}] (u2) at (4.5,1.5) {};

   \draw[gray!60, fill=gray!60] (6,4.5) ellipse (0.2cm and 0.4cm);
   \draw[gray!60, fill=gray!60] (6,3.5) ellipse (0.2cm and 0.4cm);
   \draw[gray!60, fill=gray!60] (6,2.5) ellipse (0.2cm and 0.4cm);
   \draw[gray!60, fill=gray!60] (6,1) ellipse (0.2cm and 0.4cm);
   \draw[gray!60, fill=gray!60] (6,0) ellipse (0.2cm and 0.4cm);

   \node[smallvertex] (x1) at (6,4.5) {};
   \node[smallvertex] (x2) at (6,3.5) {};
   \node[smallvertex] (x3) at (6,2.7) {};
   \node[smallvertex] (x4) at (6,2.3) {};
   \node[smallvertex] (x5) at (6,1) {};
   \node[smallvertex] (x6) at (6,0) {};

   \draw[->, thick, Green] (u1) to (x1);
   \draw[->, thick, Green] (u1) to (x2);
   \draw[->, thick, Green] (u1) to (x3);
   \draw[->, thick, Green] (u2) to (x4);
   \draw[->, thick, Green] (u2) to (x5);
   \draw[->, thick, Green] (u2) to (x6);

   \draw[decorate,decoration={brace,amplitude=10pt,mirror,raise=2pt},yshift=0pt,xshift=-4pt] (6.3,2.1) -- (6.3,4.9) node [black,midway,xshift=28pt] {\footnotesize $L_{1,i+1}^{u}$};
   \draw[decorate,decoration={brace,amplitude=10pt,mirror,raise=2pt},yshift=0pt,xshift=12pt] (6.3,-0.4) -- (6.3,2.9) node [black,midway,xshift=28pt] {\footnotesize $L_{1,i+1}^{u'}$};

   \node at (6,1.85) {$\vdots$};
   \node at (4.5,2.6) {$\vdots$};
  \end{tikzpicture}
  }
  \caption{The figure shows the sets $W_{1,i}$ (orange), $U_{1,i}$ (blue) and $L_{1,i+1}^{u}$ computed in the proof of Lemma~\ref{lem:lift-isomorphisms}.
   The color $c^*$ is shown in green and gray regions depict parts of the partition $\CQ_1$.}
  \label{fig:extend-isomorphisms}
 \end{figure}
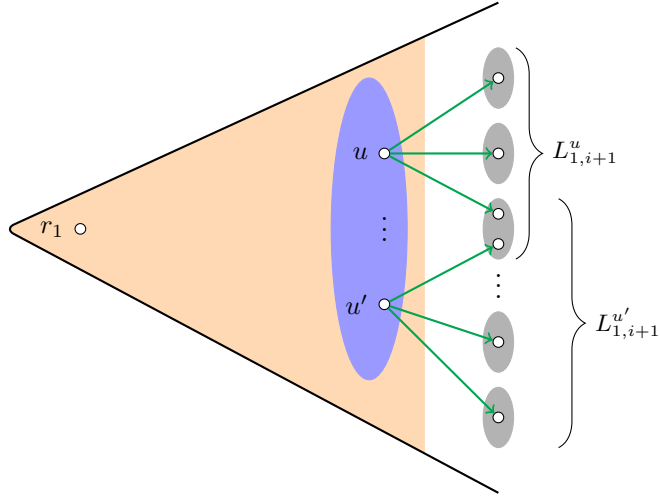

 \begin{claim}
  \label{cl:compute-isomorphism-set-local}
  For every $j,j' \in \{1,2\}$ and every $u \in U_{j,i}$, $u' \in U_{j',i}$ the set $\Iso(T_j^u,T_{j'}^{u'})$ can be computed in time $d^{O(\log d)} \cdot n^{O(1)}$.
 \end{claim}

 \begin{subproof}
  Fix $j,j' \in \{1,2\}$ and $u \in U_{j,i}$, $u' \in U_{j',i}$.
  We define $\CL \coloneqq \CL_{j,i+1}^{u} \uplus \CL_{j',i+1}^{u'}$ (i.e., we take the disjoint union of the two sets).

  For each $Q \in \CL$ we fix an orbit $A_Q \subseteq Q$ of the group $\Aut(T_b[Q])$, where $b \in \{j,j'\}$ is the unique index such that $Q \in \CQ_b$, subject to the following compatibility condition:
  If $b,b' \in \{j,j'\}$, $Q \in \CL \cap \CQ_b$, $Q' \in \CL \cap \CQ_{b'}$, and $\varphi \in \Iso(T_b[Q],T_{b'}[Q'])$ then $\varphi(A_Q) = A_{Q'}$.
  To see that this is possible, recall that if $\varphi \in \Iso(T_b[Q],T_{b'}[Q'])$ then $\varphi^{-1}\Aut(T_j[Q])\varphi = \Aut(T_{b'}[Q'])$.
  Observe that $|A_Q|$ is odd for every $Q \in \CL$.

  We define the tournament $\widetilde{T}$ with vertex set $V(\widetilde{T}) \coloneqq \CL_{j,i+1}^{u}$ and edge set
  \[E(\widetilde{T}) \coloneqq \Big\{(Q,Q') \Bigmid |E_{T_j}(A_Q,A_{Q'})| > |E_{T_j}(A_{Q'},A_Q)|\Big\}.\]
  Similarly, we define the tournament $\widetilde{T}'$ with vertex set
  $V(\widetilde{T}') \coloneqq \CL_{j',i+1}^{u'}$ and edge set
  \[E(\widetilde{T}') \coloneqq \Big\{(Q,Q') \Bigmid |E_{T_{j'}}(A_Q,A_{Q'})| > |E_{T_{j'}}(A_{Q'},A_Q)|\Big\}.\]
  Also, we color the vertices of both tournaments in such a way that vertices $Q,Q' \in \CL$ receive the same color if and only if $\Iso(T_b[Q],T_{b'}[Q']) \neq \emptyset$, where $b,b' \in \{j,j'\}$ are the unique indices such that $Q \in \CQ_b$ and $Q' \in \CQ_{b'}$.

  Observe that $|V(\widetilde{T})| \leq d$ and $|V(\widetilde{T}')| \leq d$.
  So we can compute the set $\Iso(\widetilde{T},\widetilde{T}')$ in time $d^{O(\log d)}$ by Theorem \ref{thm:ti}.
  If $\Iso(\widetilde{T},\widetilde{T}') = \emptyset$ then also $\Iso(T_j^u,T_{j'}^{u'}) = \emptyset$ and we are done.
  So suppose that $\Iso(\widetilde{T},\widetilde{T}') \neq \emptyset$, and let us write
  \[\widetilde{\Delta}\widetilde{\theta} \coloneqq \Iso(\widetilde{T},\widetilde{T}'),\]
  i.e., $\widetilde{\Delta} = \Aut(\widetilde{T})$ and $\widetilde{\theta} \in \Iso(\widetilde{T},\widetilde{T}')$.

  We define $\Delta \leq \Sym(L_{j,i+1}^{u})$ to be the permutation group containing all elements $\delta \in \Sym(L_{j,i+1}^{u})$ such that there are $\widetilde{\delta} \in \widetilde{\Delta}$ and, for every $Q \in \CL_{j,i+1}^{u}$, some $\varphi_Q \in \Iso(T_j[Q],T_j[\widetilde{\delta}(Q)])$ such that
  \[\delta(v) = \varphi_Q(v)\]
  for all $v \in L_{j,i+1}^{u}$ where $Q \in \CL_{j,i+1}^{u}$ is the unique element such that $v \in Q$.
  Observe that $\Delta$ is solvable since it is a wreath product of solvable groups (see Theorem \ref{thm:aut-solvable}).
  Moreover, a generating set for $\Delta$ can be computed in polynomial time (see Section~\ref{sec:groups} for details).

  Next, we define a bijection $\theta\colon L_{j,i+1}^{u} \to L_{j',i+1}^{u'}$ as follows.
  For each $Q \in \CL_{j,i+1}^{u}$ pick an arbitrary isomorphism $\varphi_Q \in \Iso(T_j[Q],T_{j'}[\widetilde{\theta}(Q)])$.
  Then, set
  \[\theta(v) \coloneqq \varphi_Q(v)\]
  for all $v \in L_{j,i+1}^{u}$ where $Q \in \CL_{j,i+1}^{u}$ is the unique element such that $v \in Q$.
  We have that
  \[\Iso(T_j^u,T_{j'}^{u'}) \subseteq \Delta\theta.\]
  So
  \[\Iso(T_j^u,T_{j'}^{u'}) = \Iso_{\Delta\theta}(T_j^u,T_{j'}^{u'})\]
  can be computed in polynomial time using Theorem \ref{thm:gi-solvable-group}.
 \end{subproof}

 Let $U_{i} \coloneqq U_{1,i} \cup U_{2,i}$.
 We define an equivalence relation $\sim$ on $U_i$ by setting $u \sim u'$ if $\Iso(T_j^u,T_{j'}^{u'}) \neq \emptyset$ where $j,j' \in \{1,2\}$ are the unique indices such that $u \in U_{j,i}$ and $u' \in U_{j',i}$.
 We pick an arbitrary equivalence class $U_i^* \subseteq U_i$ of the relation $\sim$.
 Let $U_{j,i}^* \coloneqq U_{j,i} \cap U_i^*$ for every $j \in \{1,2\}$.
 Also, let
 \[L_{j,i+1} \coloneqq \bigcup_{u \in U_{j,i}^*} L_{j,i+1}^u\]
 for every $j \in \{1,2\}$.

 \begin{claim}
  \label{cl:invariant-extension}
  $\varphi(U_{1,i}^*) = U_{2,i}^*$ and $\varphi(L_{1,i+1}) = L_{2,i+1}$ for every $\varphi \in \Iso((T_1,r_1),(T_2,r_2))$.
 \end{claim}

 \begin{subproof}
  Let $\varphi \in \Iso((T_1,r_1),(T_2,r_2))$.
  By \ref{item:invariant-2} we get that $\varphi(W_{1,i}) = W_{2,i}$.
  This implies that $\varphi(U_{1,i}) = U_{2,i}$.

  Let $u \in U_{1,i}^*$.
  Then $\varphi$ restricts to an isomorphism from $T_1^u$ to $T_2^{\varphi(u)}$.
  In particular, $u \sim \varphi(u)$ which implies that $\varphi(u) \in U_{2,i}^*$.
  So $\varphi(U_{1,i}^*) \subseteq U_{2,i}^*$.
  Also note that $\varphi(L_{1,i+1}^u) = L_{2,i+1}^{\varphi(u)}$.
  So $\varphi(L_{1,i+1}) \subseteq L_{2,i+1}$.

  In the other direction, let $u' \in U_{2,i}^*$.
  Then $\varphi^{-1}$ restricts to an isomorphism from $T_1^{\varphi^{-1}(u')}$ to $T_2^{u}$.
  In particular, $\varphi^{-1}(u') \sim u'$ which implies that $\varphi^{-1}(u') \in U_{1,i}^*$.
  So $U_{2,i}^* \subseteq \varphi(U_{1,i}^*)$.
  Also note that $\varphi^{-1}(L_{2,i+1}^{u'}) = L_{1,i+1}^{\varphi^{-1}(u')}$.
  So $L_{2,i+1} \subseteq \varphi(L_{1,i+1})$.
 \end{subproof}

 If $|U_{1,i}^*| \neq |U_{2,i}^*|$ or $|L_{1,i+1}| \neq |L_{2,i+1}|$ then $\Iso((T_1,r_1),(T_2,r_2)) = \emptyset$ and the algorithm terminates.
 So suppose that $|U_{1,i}^*| = |U_{2,i}^*|$ and $|L_{1,i+1}| = |L_{2,i+1}|$.
 We set
 \[W_{j,i+1} \coloneqq W_{j,i} \cup L_{j,i+1}\]
 for both $j \in \{1,2\}$.
 Clearly, \ref{item:invariant-1} and \ref{item:invariant-2} remain satisfied using Claim \ref{cl:invariant-extension}.
 It remains to compute
 \[\Iso((T_1[W_{1,i+1}],r_1),(T_2[W_{2,i+1}],r_2)).\]
 Recall that $\Gamma_i\theta_i = \Iso((T_1[W_{1,i}],r_1),(T_2[W_{2,i}],r_2))$ denotes the isomorphism set computed in the previous iteration; see Equation~\eqref{eq:def-gamma-i}.
 We start by computing the set
 \[\Gamma_i'\theta_i' \coloneqq \Iso_{\Gamma_i\theta_i}((T[W_{1,i}],r_1,U_{1,i}^*),(T[W_{2,i}],r_2,U_{2,i}^*))
 \]
 of all bijections from $\Gamma_i\theta_i$ mapping $U_{1,i}^*$ to $U_{2,i}^*$ in polynomial time using Theorem \ref{thm:hi-solvable-group}
 (recall that $\Gamma_i = \Aut(T[W_{1,i}],r_1)$ is solvable by Theorem \ref{thm:aut-solvable}).
 As before, we may assume the isomorphism set is non-empty; otherwise $\Iso((T_1,r_1),(T_2,r_2)) = \emptyset$ using Invariant \ref{item:invariant-2} and we terminate.

 Now, we consider the sets
 \[A_{j,i+1} \coloneqq \{(u,w) \mid u \in U_{j,i}^*, w \in L_{j,i+1}^{u}\}\]
 for both $j \in \{1,2\}$.
 Intuitively speaking, the reader is encouraged to view $A_{j,i+1}$ as a disjoint union over all sets $L_{j,i+1}^{u}$, $u \in U_{j,i}^*$.
 The basic idea is to first compute suitable ``isomorphisms'' between the sets $A_{1,i+1}$ and $A_{2,i+1}$.
 Afterwards, we ``merge'' elements in the sets $A_{j,i+1}$ with the same second coordinate to obtain the isomorphism set between $L_{1,i+1}$ and $L_{2,i+1}$.

 \begin{claim}
  \[\varphi(A_{1,i+1}) \coloneqq \{(\varphi(u),\varphi(w)) \mid (u,w) \in A_{1,i+1}\} = A_{2,i+1}\]
  for every $\varphi \in \Iso((T_1,r_1),(T_2,r_2))$.
 \end{claim}

 \begin{subproof}
  Let $\varphi \in \Iso((T_1,r_1),(T_2,r_2))$.
  Also let $(u,w) \in A_{1,i+1}$.
  Then $\varphi(u) \in U_{2,i}^*$ by Claim \ref{cl:invariant-extension}.
  Also, $\varphi(L_{1,i+1}^u) = L_{2,i+1}^{\varphi(u)}$.
  In particular, $\varphi(w) \in L_{2,i+1}^{\varphi(u)}$.
  So we get that $(\varphi(u),\varphi(w)) \in A_{2,i+1}$.

  In the other direction, let $(u',w') \in A_{2,i+1}$.
  Then $\varphi^{-1}(u') \in U_{1,i}^*$ by Claim \ref{cl:invariant-extension}.
  Also, $\varphi^{-1}(L_{2,i+1}^{u'}) = L_{1,i+1}^{\varphi^{-1}(u')}$.
  In particular, $\varphi^{-1}(w') \in L_{1,i+1}^{\varphi^{-1}(u')}$.
  So we get that $(\varphi^{-1}(u'),\varphi^{-1}(w')) \in A_{1,i+1}$.
 \end{subproof}

 We now define $\Delta_{i+1} \leq \Sym(A_{1,i+1})$ to be the permutation group containing all elements $\delta \in \Sym(A_{1,i+1})$ such that there is a $\gamma \in \Gamma_i'$ and, for each $u \in U_{1,i}^*$, a permutation $\varphi_u \in \Iso(T_1^u,T_1^{\gamma(u)})$ such that
 \[\delta(u,w) = (\gamma(u),\varphi_u(w))\]
 for all $(u,w) \in A_{1,i+1}$.
 Observe that $\Delta_{i+1}$ is solvable since it is a wreath product of solvable groups.
 Also, a generating set for $\Delta_{i+1}$ can be computed in polynomial time (see Section~\ref{sec:groups} for details).

 We define a bijection $\delta_{i+1}\colon A_{1,i+1} \to A_{2,i+1}$ as follows.
 For each $u \in U_{1,i}^*$ pick an arbitrary isomorphism $\varphi_u \in \Iso(T_1^u,T_2^{\theta_i'(u)})$.
 Then, set
 \[\delta_{i+1}(u,w) \coloneqq (\theta_i'(u),\varphi_u(w))\]
 for all $(u,w) \in A_{1,i+1}$.

 For $j \in \{1,2\}$ we define a hypergraph $\CH_{j,i+1} = (A_{j,i+1},\CE_{j,i+1})$ where
 \[\CE_{j,i+1} \coloneqq \{\{(u,w') \in A_{j,i+1} \mid w' = w\} \mid w \in L_{j,i+1}\}.\]
 We compute
 \[\Delta_{i+1}'\delta_{i+1}' \coloneqq \Iso_{\Delta_{i+1}\delta_{i+1}}(\CH_{1,i+1},\CH_{2,i+1})\]
 and the natural action $\Delta_{i+1}^*\delta_{i+1}^*$ of $\Delta_{i+1}'\delta_{i+1}'$ on $L_{1,i+1}$ (via the correspondence between $L_{j,i+1}$ and $\CE_{j,{i+1}}$).

 We get that
 \[\Iso((T_1[W_{1,i+1}],r_1),(T_2[W_{2,i+1}],r_2)) \subseteq \Gamma_i'\theta_i' \times \Delta_{i+1}^*\delta_{i+1}^*.\]
 Indeed, let $\varphi\in \Iso((T_1[W_{1,i+1}],r_1),(T_2[W_{2,i+1}],r_2))$.
 Then $\varphi(U_{1,i}^*) = U_{2,i}^*$ (see Claim \ref{cl:invariant-extension}) and thus $\varphi|_{W_{1,i}} \in \Gamma_i'\theta_i'$.
 Moreover, the action of $\varphi$ on $A_{1,i+1} = U_{1,i}^* \times L_{1,i+1}$ is in $\Delta_{i+1}'\delta_{i+1}'$ and thus $\varphi|_{L_{1,i+1}} \in \Delta_{i+1}^*\delta_{i+1}^*$.
 In combination, the mapping $\varphi\mapsto (\varphi|_{W_{1,i}},\varphi|_{L_{1,i+1}})$ embeds the isomorphism set $\Iso((T_1[W_{1,i+1}],r_1),(T_2[W_{2,i+1}],r_2))$ into $\Gamma_i'\theta_i' \times \Delta_{i+1}^*\delta_{i+1}^*$.

 Since $\Gamma_i' \times \Delta_{i+1}^*$ is solvable, we can compute $\Iso((T_1[W_{1,i+1}],r_1),(T_2[W_{2,i+2}],r_2))$ in polynomial time using Theorem \ref{thm:gi-solvable-group}.

 This entire procedure is repeated until we reach the stage where $W_{1,\ell} = V_1$ and $W_{2,\ell} = V_2$ (or the algorithm concludes that $\Iso((T_1,r_1),(T_2,r_2)) = \emptyset$).
 At this point the algorithm has computed $\Iso((T_1,r_1),(T_2,r_2))$ as desired.
 Note that $\ell \leq |V_1|$ since in each iteration at least one element is added, i.e., $W_{1,i} \subsetneq W_{1,i+1}$ for all $i \in [\ell-1]$.

 Also observe that all steps can be performed in polynomial time except for the subroutine from Claim \ref{cl:compute-isomorphism-set-local} which runs in time $d^{O(\log d)} \cdot n^{O(1)}$.
 Overall, this gives the desired bound on the running time of the algorithm.
\end{proof}

Building on the subroutine from Lemma \ref{lem:lift-isomorphisms}, we can now design an isomorphism test for $2$-WL-homogeneous tournaments of bounded twin width.

\begin{lemma}
 \label{lem:tww-isomorphism-unicolored}
 There is an algorithm that, given $2$-WL-homogeneous tournaments $T_1$ and $T_2$ and an integer $k \geq 1$, either concludes that $\tww(T_1) > k$ or computes $\Iso(T_1,T_2)$ in time $k^{O(\log k)} \cdot n^{O(1)}$.
\end{lemma}

\begin{proof}
 We run $2$-WL on $T_1$ and $T_2$.
 If $2$-WL distinguishes between $T_1$ and $T_2$, then we return that $\Iso(T_1,T_2) = \emptyset$.
 So suppose that $T_1 \simeq_2 T_2$.

 For $j \in \{1,2\}$ we define an arc-coloring via $\lambda_j(v,w) \coloneqq \WL{2}{T_j}(v,w)$ for every $(v,w) \in E(T) \cup \{(v,v) \mid v \in V(T)\}$.
 We write $\widehat{T}_j = (V(T_j),E(T_j),\lambda_j)$ for the corresponding arc-colored version of $T_j$.

 We apply Lemma \ref{lem:partition-sequence} to the graph $T_1$ and obtain a sequence of partitions $\{\{v\} \mid v \in V(T_1)\} = \CQ_{1,0} \prec \dots \prec \CQ_{1,\ell} = \{V(T_1)\}$ of $V(T_1)$ and a sequence of colors $c_1,\dots,c_\ell \in \{\WL{2}{T_1}(v,w) \mid (v,w) \in E(T_1)\}$ satisfying Properties \ref{item:partition-sequence-1} and \ref{item:partition-sequence-2} (or $\tww(T_1) > k$).
 Also, we define $\CQ_{2,i}$ to be the partition into the strongly connected components of $T_2[c_1,\dots,c_i]$.
 By the properties of the $2$-dimensional Weisfeiler-Leman algorithm we conclude that for every $i \in [\ell]$ and every $v \in V(T_2)$ it holds that
 \[|\{Q \in \CQ_{2,i-1} \mid \exists w \in Q \colon \WL{2}{T_2}(v,w) = c_i\}| \leq 2k+1,\]
 i.e., Property \ref{item:partition-sequence-2} is also satisfied for the sequence $\CQ_{2,0} \prec \dots \prec \CQ_{2,\ell}$.

 Now, for every $i \in \{0,\dots,\ell\}$, every $j,j' \in \{1,2\}$, and every $Q \in \CQ_{j,i}$, $Q' \in \CQ_{j',i}$ we inductively compute the set $\Iso(\widehat{T}_j[Q],\widehat{T}_{j'}[Q'])$.
 For $i = 0$ this trivial since $|Q| = |Q'| = 1$.

 So suppose $i \in \{0,\dots,\ell-1\}$, $j,j' \in \{1,2\}$, and $Q \in \CQ_{j,i+1}$, $Q' \in \CQ_{j',i+1}$.
 Let $d \coloneqq 2k+1$.
 We set $\widetilde{T}_1 \coloneqq \widehat{T}_j[Q]$ and $\widetilde{T}_2 \coloneqq \widehat{T}_{j'}[Q']$.
 Also, we set
 \[\widetilde{\CQ}_1 \coloneqq \{\widetilde{Q} \in \CQ_{j,i} \mid \widetilde{Q} \subseteq Q\}\]
 and
 \[\widetilde{\CQ}_2 \coloneqq \{\widetilde{Q} \in \CQ_{j',i} \mid \widetilde{Q} \subseteq Q'\}.\]
 Note that $\widetilde{\CQ}_1$ is a partition of $V(\widetilde{T}_1)$ and $\widetilde{\CQ}_2$ is a partition of $V(\widetilde{T}_2)$.
 Finally, we set $c^* \coloneqq c_{i+1}$ and apply Lemma \ref{lem:lift-isomorphisms} to compute the set
 \[\Iso(\widehat{T}_j[Q],\widehat{T}_{j'}[Q']) = \Iso(\widetilde{T}_1,\widetilde{T}_2)\]
 in time $d^{O(\log d)} \cdot n^{O(1)} = k^{O(\log k)} \cdot n^{O(1)}$.
 Observe that all isomorphism sets required in Lemma \ref{lem:lift-isomorphisms}\ref{item:lift-isomorphisms-5} have already been computed in the previous iteration.

 Finally, observe that $\CQ_{1,\ell} = \{V(T_1)\}$ and $\CQ_{2,\ell} = \{V(T_2)\}$.
 So in iteration $i = \ell$ the algorithm computes
 \[\Iso(\widehat{T}_1,\widehat{T}_2) = \Iso(T_1,T_2).\]

 To bound the running algorithm, one can observe that all steps can be performed in polynomial time except for the subroutine in Lemma \ref{lem:lift-isomorphisms}.
 However, this subroutine is only called at most $\ell \cdot (|V(T_1)| + |V(T_2)|)^2$ many times, which is polynomial in the input size, and runs in time $k^{O(\log k)} \cdot n^{O(1)}$.
 Overall, this gives the desired running time.
\end{proof}

Finally, relying on well-established arguments, we can remove the restriction that all vertices need to receive the same color under $2$-WL to obtain Theorem \ref{thm:tww-isomorphism-intro}.
We reformulate our main algorithmic result in a slightly more precise manner.

\begin{theorem}
 \label{thm:tww-isomorphism}
 There is an algorithm that, given two tournaments $T_1$ and $T_2$ and an integer $k \geq 1$, either concludes that $\tww(T_1) > k$ or computes $\Iso(T_1,T_2)$ in time $k^{O(\log k)} \cdot n^{O(1)}$.
\end{theorem}

\begin{proof}
 We run $2$-WL on $T_1$ and $T_2$.
 If $2$-WL distinguishes between $T_1$ and $T_2$, then we return that $\Iso(T_1,T_2) = \emptyset$.
 So suppose that $T_1 \simeq_2 T_2$.

 Let $C_V \coloneqq \{\WL{2}{T_1}(v,v) \mid v \in V(T_1)\}$ denote the set of \emph{vertex colors} of the coloring $\WL{2}{T_1}$.
 For $c \in C_V$ and $j \in \{1,2\}$ we define
 \[X_{j,c} \coloneqq \{v \in V(T_j) \mid \WL{2}{T_j}(v,v) = c\}.\]
 For every $c \in C_V$ we use the algorithm from Lemma \ref{lem:tww-isomorphism-unicolored} to compute the set
 \[\Gamma_c\theta_c \coloneqq \Iso(T_1[X_{1,c}],T_2[X_{2,c}]).\]
 If there is some $c \in C_V$ for which the algorithm from Lemma \ref{lem:tww-isomorphism-unicolored} concludes that $\tww(T_1[X_{1,c}]) > k$, then $\tww(T_1) > k$ by Lemma \ref{lem:tww-hereditary}.

 Moreover, if there is a color $c \in C_V$ such that $\Iso(T_1[X_{1,c}],T_2[X_{2,c}]) = \emptyset$, then $\Iso(T_1,T_2) = \emptyset$.
 Otherwise, let
 \[\Gamma \coloneqq \bigtimes_{c \in C_V} \Gamma_c \leq \Sym(V(T_1))\]
 and $\theta\colon V(T_1) \to V(T_2)$ be the bijection defined via $\theta(v) \coloneqq \theta_c(v)$ for the unique $c \in C_V$ such that $v \in X_{1,c}$.
 Then
 \[\Iso(T_1,T_2) \subseteq \Gamma\theta.\]
 Moreover, the group $\Gamma_c$ is solvable for every $c \in C_V$ by Theorem \ref{thm:aut-solvable}.
 This implies that $\Gamma$ is solvable.
 So
 \[\Iso(T_1,T_2) = \Iso_{\Gamma\theta}(T_1,T_2)\]
 can be computed in polynomial time using Theorem \ref{thm:gi-solvable-group}.
\end{proof}

\section{The WL-Dimension of Tournaments of Bounded Twin Width}
\label{sec:wl}

In this section, we prove Theorem \ref{thm:wl-tournament-tww}.
The proof relies on well-known game-characterizations of the Weisfeiler-Leman algorithm as well as the graph parameter tree width.

\subsection{Games}

First, we cover the cops-and-robber game that characterizes tree width of graphs.
Fix some integer $k \geq 1$.
For an undirected graph $G$, we define the cops-and-robber game $\CopRob_k(G)$ as follows:

\begin{itemize}
 \item The game has two players called \emph{Cops} and \emph{Robber}.
 \item The game proceeds in rounds, each of which is associated with a pair of positions
 $(\bar v,u)$ with~$\bar v \in \big(V(G)\big)^k$ and~$u \in V(G)$.
 \item To determine the initial position, the Cops first choose a tuple $\bar v = (v_1,\dots,v_k) \in \big(V(G)\big)^k$ and then the Robber chooses some vertex $u \in V(G) \setminus \{v_1,\dots,v_k\}$ (if no such $u$ exists, the Cops win the play).
  The initial position of the game is then set to $(\bar v,u)$.
 \item Each round consists of the following steps.
  Suppose the current position of the game is $(\bar v,u) = ((v_1,\dots,v_k),u)$.
  \begin{itemize}
   \item[(C)] The Cops choose some $i \in [k]$ and $v' \in V(G)$.
   \item[(R)] The Robber chooses a vertex $u' \in V(G)$ such that there exists a path from $u$ to $u'$ in the graph $G \setminus \{v_1,\dots,v_{i-1},v_{i+1},\dots,v_k\}$.
    After that, the game moves to the new position $\big((v_1,\dots,v_{i-1},v',v_{i+1},\dots,v_k), u'\big)$.
  \end{itemize}

  If $u \in \{v_1,\dots,v_k\}$ the Cops win.
  If there is no position of the play such that the Cops win, then the Robber wins.
\end{itemize}

We say that the Cops (and the Robber, respectively) \emph{win $\CopRob_k(G)$} if the Cops (and the Robber, respectively) have a winning strategy for the game.
We also say that \emph{$k$ cops can catch a robber on $G$} if the Cops have a winning strategy in this game.

\begin{theorem}[\cite{SeymourT93}]
 \label{thm:cops-and-robbers-characterization-of-treewidth}
 An undirected graph $G$ has tree width at most $k$ if and only if $k+1$ cops can catch a robber on $G$.
\end{theorem}

Next, we discuss a game-theoretic characterization of the Weisfeiler-Leman algorithm.

Let $k \geq 1$.
For graphs $G$ and $H$ on the same number of vertices, we define the \emph{bijective $k$-pebble game} $\BP_{k}(G,H)$ as follows:
\begin{itemize}
 \item The game has two players called \emph{Spoiler} and \emph{Duplicator}.
 \item The game proceeds in rounds, each of which is associated with a pair of positions
 $(\bar v,\bar w)$ with~$\bar v \in \big(V(G)\big)^k$ and~$\bar w \in \big(V(H)\big)^k$.
 \item To determine the initial position, Duplicator plays a bijection $f\colon \big(V(G)\big)^k \rightarrow \big(V(H)\big)^k$ and Spoiler chooses some $\bar v \in \big(V(G)\big)^k$.
  The initial position of the game is then set to $(\bar v,f(\bar v))$.
 \item Each round consists of the following steps.
  Suppose the current position of the game is $(\bar v,\bar w) = ((v_1,\ldots,v_k),(w_1,\ldots,w_k))$.
  \begin{itemize}
   \item[(S)] Spoiler chooses some $i \in [k]$.
   \item[(D)] Duplicator picks a bijection $f\colon V(G) \rightarrow V(H)$.
   \item[(S)] Spoiler chooses $v \in V(G)$ and sets $w \coloneqq f(v)$.
    Then the game moves to position $\big(\bar v[i/v], \bar w[i/w]\big)$ where $\bar v[i/v] \coloneqq (v_1,\dots,v_{i-1},v,v_{i+1},\dots,v_k)$ is the tuple obtained from $\bar v$ by replacing the $i$-th entry by $v$.
  \end{itemize}

  If mapping each $v_i$ to $w_i$ does not define an isomorphism of the induced subgraphs of $G$ and $H$, Spoiler wins the play.
  More precisely, Spoiler wins if there are~$i,j\in [k]$ such that~$v_i = v_j \nLeftrightarrow w_i =w_j$ or~$v_iv_j \in E(G) \nLeftrightarrow w_iw_j \in E(H)$.
  If there is no position of the play such that Spoiler wins, then Duplicator wins.
\end{itemize}

We say that Spoiler (and Duplicator, respectively) \emph{wins $\BP_k(G,H)$} if Spoiler (and Duplicator, respectively) has a winning strategy for the game.
Also, for a position $(\bar v,\bar w)$ with $\bar v \in \big(V(G)\big)^k$ and $\bar w \in \big(V(H)\big)^k$, we say that Spoiler (and Duplicator, respectively) \emph{wins $\BP_k(G,H)$ from position $(\bar v,\bar w)$} if Spoiler (and Duplicator, respectively) has a winning strategy for the game started at position $(\bar v,\bar w)$.

\begin{theorem}[\cite{CaiFI92,Hella96}]
 \label{thm:eq-wl-pebble}
 Suppose $k \geq 2$ and let $G$ and $H$ be two directed graphs.
 Then $G \simeq_k H$ if and only if Duplicator wins the game $\BP_{k+1}(G,H)$.
\end{theorem}

\subsection{CFI Graphs and Tournaments}

To prove that the WL algorithm on its own is unable to determine isomorphisms between tournaments of bounded twin width, we adapt a construction of Cai, F{\"{u}}rer and Immerman \cite{CaiFI92}.
Towards this end, we first describe a construction of directed graphs with large WL dimension, and then argue how to translate those graphs into tournaments while preserving their WL dimension.

More concretely, the Cai-F{\"{u}}rer-Immerman (CFI) construction \cite{CaiFI92} takes an undirected base graph $G$, and replaces every vertex with a certain gadget, and those gadgets are connected along the edges of $G$.
To obtain a lower bound on the WL dimension that is linear in the number of vertices, one typically uses $3$-regular expander graphs as a base graph (more precisely, we require a bounded-degree graph $G$ with tree width that is linear in the number of vertices).
In order to translate this construction to tournaments of small twin-width, we need to tackle some hurdles.
First of all, in the standard CFI construction, each gadget encodes a linear equation over $\ZZ_2$ via its automorphism group.
However, the automorphism group of a tournament always has odd order (see the comments before Theorem \ref{thm:aut-solvable}).
For this reason, we rely on a simple variant of the gadgets that encode linear equations over $\ZZ_3$.
At this point, we can employ a simple idea to translate a CFI graph into a tournament:
We fix a linear order $<$ on the vertex set of $G$ (to be specified later), and represent edges of the CFI graph by a ``forward'' edge in the tournament (i.e., an edge that points to the larger element in the linear order) whereas non-edges are represented by ``backward'' edges.

To bound the twin width of the resulting tournament, the basic idea to build the contraction sequence is to first contract each of the CFI gadgets to a single vertex.
At this point, intuitively speaking, we obtain (the tournament representation of) the trigraph $G^{\red}$ obtained from $G$ by replacing every edge with a red edge.
So, in order to bound the twin width, we choose $G$ to be a $3$-regular expander graph such that $G^{\red}$ has small twin width; the existence of such graphs is shown in \cite{BonnetGKTW22}.
Finally, for the contraction sequence of $G^{\red}$ to be compatible with the tournament representation described above, we choose the linear order $<$ so that the twin width of $(G^{\red},<)$ is small; this is possible by Lemma~\ref{lem:twin-width-order}.

In the following, we formally describe the construction of tournaments with large WL dimension.
Let $G$ be a connected, $3$-regular undirected graph.
Let $G^{\red}$ denote the structure obtained from $G$ be replacing every edge with a red edge and let $t \coloneqq \tww(G^{\red})$.
By Lemma \ref{lem:twin-width-order} there is some linear order $<$ on $V(G)$ such that
\begin{equation}
 \label{eq:order-twin-width}
 \tww(G^\red,<) = t.
\end{equation}
Also, let $\vec{G}$ be an arbitrary orientation of $G$.

We start by providing the basic CFI construction (where gadgets encode linear equations over $\ZZ_3$).
We stress that the specific construction described below is heavily tailored towards enabling the later translation into tournaments.

Recall that for every $v \in V(G)$ we denote by $E_+(v)$ the set of outgoing edges and $E_-(v)$ the set of incoming edges in $\vec{G}$.
Also, we write $E(v)$ to denote the set of incident (undirected) edges in $G$.
For $a \in \ZZ_3$ we define
\[M_a(v) \coloneqq \Big\{f\colon E(v) \to \ZZ_3 \Bigmid \sum_{(v,w) \in E_+(v)} f(\{v,w\}) - \sum_{(w,v) \in E_-(v)} f(\{v,w\}) = a \pmod{3}\Big\}\]
We also define $F_a(v)$ to contain all pairs $(f,g) \in (M_a(v))^2$ such that, for the minimal element $w \in N_G(v)$ (with respect to $<$) such that $f(vw) \neq g(vw)$, it holds that
\[f(vw) + 1 = g(vw) \pmod{3}.\]
Observe that, for every distinct $f,g \in M_a(v)$, either $(f,g) \in F_a(v)$ or $(g,f) \in F_a(v)$.

Let $\alpha\colon V(G) \to \ZZ_3$ be a function.
We define the graph $\CFI_3(\vec{G},<,\alpha)$ with vertex set
\[V(\CFI_3(\vec{G},<,\alpha)) \coloneqq \bigcup_{v \in V(G)} \{v\} \times M_{\alpha(v)}(v)\]
and edge set
\begin{align*}
 E(\CFI_3(\vec{G},<,\alpha)) \coloneqq \quad  &\Big\{\{(v,f),(w,g)\} \Bigmid vw \in E(G) \wedge f(vw) = g(vw)\Big\}\\
                                       \cup\; &\Big\{((v,f),(v,g)) \Bigmid (f,g) \in F_{\alpha(v)}(v)\Big\}.
\end{align*}
Observe that $\CFI_3(\vec{G},<,\alpha)$ is a \emph{mixed graph}, i.e., it contains both directed and undirected edges.
Also, we color the vertices of $\CFI_3(\vec{G},<,\alpha)$ using the coloring $\lambda\colon V(\CFI_3(\vec{G},<,\alpha)) \to V(G)$ defined via $\lambda(v,f) \coloneqq v$ for all $(v,f) \in V(\CFI_3(\vec{G},<,\alpha))$, i.e., each set $\{v\} \times M_{\alpha(v)}(v)$ forms a color class under $\lambda$.

Let us analyze the relevant properties of the graphs $\CFI_3(\vec{G},<,\alpha)$.
We start by investigating isomorphisms between the graphs $\CFI_3(\vec{G},<,\alpha)$ for different mappings $\alpha\colon V(G) \to \ZZ_3$.

\begin{lemma}
 \label{lem:twist-along-path}
 Let $\alpha\colon V(G) \to \ZZ_3$ be a function and let $v,w \in V(G)$ be distinct vertices.
 Also let $P$ be a path from $v$ to $w$ in $G$.
 Let $\beta\colon V(G) \to \ZZ_3$ denote the function defined via
 \[\beta(u) \coloneqq \begin{cases}
                       \alpha(u) + 1 \pmod{3} &\text{if } u = v\\
                       \alpha(u) - 1 \pmod{3} &\text{if } u = w\\
                       \alpha(u)              &\text{otherwise}
                      \end{cases}.\]
 Then there is an isomorphism $\varphi\colon \CFI_3(\vec{G},<,\alpha) \cong \CFI_3(\vec{G},<,\beta)$ such that
 \[\varphi(u,f) = (u,f)\]
 for all $u \in V(G) \setminus V(P)$ and $f \in M_{\alpha(u)}(u)$.
\end{lemma}

\begin{proof}
 Let $E(P)$ denote the set of directed edges on the path $P$ (all edges are directed so that we obtain a path from $v$ to $w$).
 For every $u \in V(G)$ we define the function $h_{u}\colon E(u) \to \ZZ_3$ via
 \[h_u(\{u,u'\}) = \begin{cases}
                    1 &\text{if } (u,u') \in E(P) \wedge (u,u') \in E(\vec{G})\\
                    2 &\text{if } (u,u') \in E(P) \wedge (u',u) \in E(\vec{G})\\
                    2 &\text{if } (u',u) \in E(P) \wedge (u,u') \in E(\vec{G})\\
                    1 &\text{if } (u',u) \in E(P) \wedge (u',u) \in E(\vec{G})\\
                    0 &\text{otherwise}
                   \end{cases}.\]

 We define $\varphi(u,f) \coloneqq (u,f + h_u)$ for all $(u,f) \in \CFI_3(\vec{G},<,\alpha)$.
 It is easy to check that $\varphi\colon \CFI_3(\vec{G},<,\alpha) \cong \CFI_3(\vec{G},<,\beta)$ and the desired properties are satisfied.
\end{proof}

\begin{corollary}
 \label{cor:cfi-isomorphism}
 Let $\alpha,\beta\colon V(G) \to \ZZ_3$ be two functions such that
 \[\sum_{v \in V(G)} \alpha(v) = \sum_{v \in V(G)} \beta(v) \pmod{3}.\]
 Then $\CFI_3(\vec{G},<,\alpha) \cong \CFI_3(\vec{G},<,\beta)$.
\end{corollary}

\begin{proof}
 This follows from repeatedly applying Lemma \ref{lem:twist-along-path} to gradually align the two functions $\alpha$ and $\beta$.
\end{proof}

Now fix an arbitrary vertex $u_0 \in V(G)$.
For every $p \in \ZZ_3$ we define the mapping $\alpha_p\colon V(G) \to \ZZ_3$ via $\alpha_p(u_0) \coloneqq p$ and $\alpha_p(w) \coloneqq 0$ for all $w \in V(G) \setminus \{u_0\}$.

\begin{lemma}
 \label{lem:cfi-non-isomorphism}
 $\CFI_3(\vec{G},<,\alpha_0) \not\cong \CFI_3(\vec{G},<,\alpha_1)$.
\end{lemma}

\begin{proof}
 Let $X \coloneqq \{(v,\mathbf{0}_{E(v)}) \mid v \in V(G)\}$ where $\mathbf{0}_{E(v)}$ denotes the function that maps every element in $E(v)$ to $0$.
 Note that $(\CFI_3(\vec{G},<,\alpha_0))[X]$ is isomorphic to $G$ and contains exactly one vertex from each color class.

 Now suppose towards a contradiction that $\varphi\colon\CFI_3(\vec{G},<,\alpha_0) \cong \CFI_3(\vec{G},<,\alpha_1)$.
 Let $Y \coloneqq \varphi(X)$.
 Since $\varphi$ is an isomorphism, the set $Y$ also contains exactly one vertex from every color class, i.e., for every $v \in V(G)$ there is a unique $f_v \in M_{\alpha_1(v)}(v)$ such that $(v,f_v) \in Y$.
 Also, $f_v(\{v,w\}) = f_w(\{v,w\})$ for all edges $\{v,w\} \in E(G)$ because $\{(v,f_v),(w,f_w)\} \in E(\CFI_3(\vec{G},<,\alpha_1))$.
 This means we can define combine all the mappings $f_v$, $v \in V(G)$, into one assignment $f \colon E(G) \to \ZZ_3$ via
 \[f(\{v,w\}) \coloneqq f_v(\{v,w\}) = f_w(\{v,w\}).\]
 For every $v \in V(G)$ it holds that
 \[\sum_{(v,w) \in E_+(v)} f(\{v,w\}) - \sum_{(w,v) \in E_-(v)} f(\{v,w\}) = \alpha_1(v) \pmod{3}.\]
 So
 \[\sum_{v \in V(G)} \left(\sum_{(v,w) \in E_+(v)} f(\{v,w\}) - \sum_{(w,v) \in E_-(v)} f(\{v,w\})\right) = \sum_{v \in V(G)}\alpha_1(v) \pmod{3}.\]
 However, we have that
 \[\sum_{v \in V(G)} \left(\sum_{(v,w) \in E_+(v)} f(\{v,w\}) - \sum_{(w,v) \in E_-(v)} f(\{v,w\})\right) = 0 \pmod{3}\]
 whereas
 \[\sum_{v \in V(G)}\alpha_1(v) = 1 \pmod{3}.\]
 This is a contradiction.
\end{proof}

Now, we analyze the WL algorithm on the graphs $\CFI_3(\vec{G},<,\alpha)$ for different functions $\alpha\colon V(G) \to \ZZ_3$.
Similar results have been obtained in prior works; see, e.g., \cite{DawarR07,Neuen24}.

\begin{lemma}
 \label{lem:cfi-indistinguishable}
 Let $k$ be an integer such that $\tw(G) \geq k+1$.
 Also let $\alpha,\beta\colon V(G) \to \ZZ_3$ be two functions.
 Then $\CFI_3(\vec{G},<,\alpha) \simeq_k \CFI_3(\vec{G},<,\beta)$.
\end{lemma}

\begin{proof}
 Using Corollary \ref{cor:cfi-isomorphism} it suffices to show that $\CFI_3(\vec{G},<,\alpha_p) \simeq_k \CFI_3(\vec{G},<,\alpha_q)$ for all $p,q \in \ZZ_3$.
 In fact, since $\simeq_k$ is an equivalence relation, we only need to prove that $\CFI_3(\vec{G},<,\alpha_0) \simeq_k \CFI_3(\vec{G},<,\alpha_q)$ for all $q \in \{1,2\}$.

 So let us fix $q \in \{1,2\}$.
 Since $\tw(G) \geq k+1$, the Robber has a winning strategy in the cops-and-robber game $\CopRob_{k+1}(G)$ by Theorem \ref{thm:cops-and-robbers-characterization-of-treewidth}.
 We translate the winning strategy for the Robber in $\CopRob_{k+1}(G)$ into a winning strategy for Duplicator in the $(k+1)$-bijective pebble game played on $\CFI_3(\vec{G},<,\alpha_0)$ and $\CFI_3(\vec{G},<,\alpha_q)$.
 Using Theorem \ref{thm:eq-wl-pebble} this implies that $\CFI_3(\vec{G},<,\alpha_0) \simeq_k \CFI_3(\vec{G},<,\alpha_q)$ as desired.

 For a tuple $\bar x = (x_1,\dots,x_{k+1}) \in (V(\CFI_3(\vec{G},<,\alpha_0)))^{k+1}$, we define $A(\bar x) \coloneqq (v_1,\dots,v_{k+1})$ where $v_i \in V(G)$ is the unique vertex such that $x_i \in \{v_i\} \times M_{0}(v_i)$.

 We first construct the bijection $f$ for the initialization round.
 Let $\bar v = (v_1,\dots,v_{k+1}) \in (V(G))^{k+1}$, and consider a tuple $\bar x = (x_1,\dots,x_{k+1}) \in (V(\CFI_3(\vec{G},<,\alpha_0)))^{k+1}$ with $A(\bar x) = \bar v$.
 Now let $u$ be the vertex chosen by the Robber if the Cops initially place themselves on $A(\bar x)$.
 We define $\alpha_{u,q}\colon V(G) \to \ZZ_3$ via $\alpha_{u,q}(u) \coloneqq q$ and $\alpha_{u,q}(w) \coloneqq 0$ for all $w \in V(G) \setminus \{u\}$.
 Let $P$ be a path from $u$ to $u_0$ (recall that $G$ is connected), and let $\varphi$ denote the isomorphism from $\CFI_3(\vec{G},<,\alpha_{u,q})$ to $\CFI_3(\vec{G},<,\alpha_{q})$ constructed in Lemma \ref{lem:twist-along-path}.
 We set $f(\bar x) \coloneqq (\varphi(x_1),\dots,\varphi(x_{k+1}))$.
 It is easy to see that this gives a bijection $f$ (we use the same isomorphism $\varphi$ for all tuples $\bar x$ such that $A(\bar x) = \bar v$).
 Observe that $x_i \in V(\CFI_3(\vec{G},<,\alpha_{u,q}))$ since $\alpha_{u,q}(v_i) = 0 = \alpha_0(v_i)$ for all $i \in [k+1]$.

 Now, throughout the game, Duplicator maintains the following invariant.
 Let $(\bar x,\bar y)$ denote the current position.
 Then there is a vertex $u \in V(G)$ and a bijection $\varphi\colon \CFI_3(\vec{G},<,\alpha_{u,q}) \to \CFI_3(\vec{G},<,\alpha_{q})$ such that
 \begin{itemize}
  \item $\varphi(\bar x) = \bar y$,
  \item $u$ does not appear in the tuple $A(\bar x)$, and
  \item the Robber wins from the position $(A(\bar x), u)$, i.e., if the Cops are placed on $A(\bar x)$ and the Robber is on $u$.
 \end{itemize}
 Note that this condition is satisfied by construction after the initialization round.
 Also observe that Duplicator never looses the game in such a position.
 Indeed, since $u$ does not appear in the tuple $A(\bar x)$, the graphs $\CFI_3(\vec{G},<,\alpha_0)$ and $\CFI_3(\vec{G},<,\alpha_{u,q})$ are identical when restricted to vertices from $\bar x$.
 Since $\varphi$ is an isomorphism from $\CFI_3(\vec{G},<,\alpha_{u,q})$ to $\CFI_3(\vec{G},<,\alpha_{q})$ and $\varphi(\bar x) = \bar y$, Duplicator does not loose in the position $(\bar x,\bar y)$.

 So it remains to show that Duplicator can maintain the above invariant in each round of the $(k+1)$-bijective pebble game.
 Suppose $(\bar x,\bar y)$ is the current position.
 Also let $(A(\bar x), u)$ be the associated position in the cops-and-robber game.
 Suppose that $A(\bar x) = (v_1,\dots,v_{k+1})$.

 Let $i \in [k+1]$ denote the index chosen by Spoiler.
 We describe the bijection $f$ chosen by Duplicator.
 Let $v \in V(G)$.
 Let $u'$ be the vertex the Robber moves to if the Cops choose $i$ and $v$ (i.e., the $i$-th cop changes its position to $v$) in the position $(A(\bar x),u)$.
 Let $P$ denote a path from $u$ to $u'$ that avoids $\{v_1,\dots,v_{k+1}\} \setminus \{v_i\}$.
 Let $\psi$ denote the isomorphism from $\CFI(G,\alpha_{u',q})$ to $\CFI(G,\alpha_{u,q})$ constructed in Lemma \ref{lem:twist-along-path}.
 We set $f(x) \coloneqq \varphi(\psi(x))$ for all $x \in \{v\} \times M_{0}(v)$.

 It is easy to see that $f$ is a bijection.
 Let $x$ denote the vertex chosen by Spoiler and let $y \coloneqq f(x)$.
 Let $\bar x' \coloneqq \bar x[i/x]$ and $\bar y' \coloneqq \bar y[i/y]$, i.e., the pair $(\bar x',\bar y')$ is the new position of the game.
 Also, we set $\varphi' \coloneqq \psi \circ \varphi$ where $\psi$ denotes the isomorphism from $\CFI(G,\alpha_{u',q})$ to $\CFI(G,\alpha_{u,q})$ used in the definition of $f(x)$.

 We have $\varphi'(x) = y$ by definition.
 All the other entries of $\bar x'$ are fixed by the mapping $\psi$ (see Lemma \ref{lem:twist-along-path}) which overall implies that $\varphi'(\bar x') = \bar y'$.
 Also, $u'$ does not appear in the tuple $A(\bar x')$ by construction, and the Robber wins from the position $(A(\bar x'), u')$.

 So overall, this means that Duplicator can maintain the above invariant which provides the desired winning strategy.
\end{proof}

Together, Lemmas \ref{lem:cfi-non-isomorphism} and \ref{lem:cfi-indistinguishable} provide pairs of non-isomorphic graphs that are not distinguished by $k$-WL assuming $\tw(G) > k$.
Next, we argue how to turn these graphs into tournaments.

We define the tournament $T = T(\vec{G},<,\alpha)$ with vertex set
\[V(T) \coloneqq V(\CFI_3(\vec{G},<,\alpha)) = \bigcup_{v \in V(G)} \{v\} \times M_{\alpha(v)}(v)\]
and edge set
\begin{align*}
 E(T) \coloneqq \quad  &\Big\{((v,f)(v,g)) \Bigmid (f,g) \in F_{\alpha(v)}(v)\Big\}\\
                \cup\; &\Big\{((v,f)(w,g)) \Bigmid v < w \wedge \{(v,f)(w,g)\} \notin E(\CFI_3(\vec{G},<,\alpha))\Big\}\\
                \cup\; &\Big\{((w,g)(v,f)) \Bigmid v < w \wedge \{(v,f)(w,g)\} \in E(\CFI_3(\vec{G},<,\alpha))\Big\}.
\end{align*}

We argue that the relevant properties are preserved by this translation.

\begin{lemma}
 \label{lem:cfi-tournament-isomorphism}
 Let $\alpha,\beta\colon V(G) \to \ZZ_3$ be two functions.
 Then $T(\vec{G},<,\alpha) \cong T(\vec{G},<,\beta)$ if and only if $\CFI_3(\vec{G},<,\alpha) \cong \CFI_3(\vec{G},<,\beta)$.
\end{lemma}

\begin{proof}
 The backward direction is immediately clear since $T(\vec{G},<,\alpha)$ and $T(\vec{G},<,\beta)$ are defined from $\CFI_3(\vec{G},<,\alpha)$ and $\CFI_3(\vec{G},<,\beta)$, respectively, in an isomorphism-invariant manner.
 Note that the linear order $<$ on the vertex set of $G$ is encoded into $\CFI_3(\vec{G},<,\alpha)$ and $\CFI_3(\vec{G},<,\beta)$ via the vertex-coloring.

 So let $\varphi\colon T(\vec{G},<,\alpha) \cong T(\vec{G},<,\beta)$.
 The following claim provides the key tool.

 \begin{claim}
  For every $(v,f) \in V(T(\vec{G},<,\alpha))$ there is some $g \in M_{\beta(v)}(v)$ such that $\varphi(v,f) = (v,g)$.
 \end{claim}
 \begin{subproof}
  Suppose $V(G) = \{v_1,\dots,v_n\}$ where $v_i < v_j$ for all $i < j \in [n]$.
  For $i \leq j$ we define
  \[X_{i,j} \coloneqq \bigcup_{i \leq \ell \leq j} \{v_\ell\} \times M_{\alpha(v_\ell)}(v_\ell)\]
  and
  \[Y_{i,j} \coloneqq \bigcup_{i \leq \ell \leq j} \{v_\ell\} \times M_{\beta(v_\ell)}(v_\ell).\]
  Also let $X_i \coloneqq X_{i,i}$ and $Y_i \coloneqq Y_{i,i}$.
  Let $T_\alpha \coloneqq T(\vec{G},<,\alpha)$ and $T_\beta \coloneqq T(\vec{G},<,\beta)$.

  We argue by induction that $\varphi(X_i) = Y_i$ for all $i \in [n]$.
  So fix some $i \in [n]$ and suppose that $\varphi(X_j) = Y_j$ for all $j < i$.
  In particular, $\varphi(X_{i,n}) = Y_{i,n}$.

  Let $w \in X_i$.
  We have
  \[|N_-(w) \cap X_{i,n}| = 4 + 3 \cdot |N_G(v_i) \cap \{v_{i+1},\dots,v_n\}|.\]
  Let $\ell \in \{i,\dots,n\}$ such that $\varphi(w) \in Y_\ell$.
  Then
  \[|N_-(\varphi(w)) \cap Y_{i,n}| = 4 + (\ell - i) \cdot 9 - 3 \cdot |N_G(v_\ell) \cap \{v_{i},\dots,v_{\ell-1}\}| +  3 \cdot |N_G(v_\ell) \cap \{v_{\ell+1},\dots,v_n\}|.\]
  Moreover, using that $\varphi(X_j) = Y_j$ for all $j < i$, we get that
  \[N_G(v_i) \cap \{v_1,\dots,v_{i-1}\} = N_G(v_\ell) \cap \{v_1,\dots,v_{i-1}\}\]
  which implies that
  \[|N_G(v_i) \cap \{v_i,\dots,v_n\}| = |N_G(v_\ell) \cap \{v_i,\dots,v_n\}|\]
  Also
  \[|N_-(w) \cap X_{i,n}| = |N_-(\varphi(w)) \cap Y_{i,n}|.\]
  All this is only possible if $\ell = i$ using that $G$ is $3$-regular.
  So $\varphi(w) \in Y_i$ as desired.
 \end{subproof}

 The claim immediately implies that $\varphi\colon\CFI_3(\vec{G},<,\alpha) \cong \CFI_3(\vec{G},<,\beta)$.
\end{proof}

\begin{lemma}
 \label{lem:cfi-tournament-indistinguishable}
 Let $k \geq 2$ be an integer such that $\tw(G) \geq k+1$.
 Also let $\alpha,\beta\colon V(G) \to \ZZ_3$ be two functions.
 Then $T(\vec{G},<,\alpha) \simeq_k T(\vec{G},<,\beta)$.
\end{lemma}

\begin{proof}
 Let $(\bar v,\bar w)$ be a winning position for Spoiler in the $(k+1)$-bijective pebble game played on $T(\vec{G},<,\alpha)$ and $T(\vec{G},<,\beta)$.
 Then $(\bar v,\bar w)$ is also a winning position for Spoiler in the $(k+1)$-bijective pebble game played on $\CFI_3(\vec{G},<,\alpha)$ and $\CFI_3(\vec{G},<,\beta)$.
 So the lemma follows directly from Lemma \ref{lem:cfi-indistinguishable}.
\end{proof}

To prove Theorem \ref{thm:wl-tournament-tww}, we also need to bound the twin width of the resulting graph.
Recall that $t \coloneqq \tww(G^{\red})$ where $G^{\red}$ denotes the version of $G$ where every edge is colored red.
Also recall that we defined the linear order $<$ on $V(G)$ in such a way that $t = \tww(G^{\red}) = \tww(G^{\red},<)$ (see Equation \eqref{eq:order-twin-width}).

\begin{lemma}
 \label{lem:tww-cfi-tournament}
 For every function $\alpha\colon V(G) \to \ZZ_3$ it holds that $\tww(T(\vec{G},<,\alpha)) \leq \max(35,t)$.
\end{lemma}

\begin{proof}
 Throughout the proof, we define $M(v) \coloneqq M_{\alpha(v)}(v)$ for every $v \in V(G)$.
 Since $G$ is $3$-regular, we have that $|M(v)| = 9$ for every $v \in V(G)$.
 So $|V(T(\vec{G},<,\alpha))| = 9 \cdot |V(G)|$.
 Also note that $(M(v),M(w))$ is homogeneous for all distinct $v,w \in V(G)$ such that $\{v,w\} \notin E(G)$.

 We construct a partial contraction sequence as follows.
 Let $n \coloneqq |V(G)|$.
 We define $\CP_1,\dots,\CP_{8n+1}$ arbitrarily such that $\CP_{8n+1} = \{M(v) \mid v \in V(G)\}$.
 Since $G$ is $3$-regular and $|M(v)| = 9$ for every $v \in V(G)$, we conclude that $(T(\vec{G},<,\alpha))/\CP_i$ has red degree at most $4 \cdot 9 - 1 = 35$ for every $i \in [8n+1]$.
 Now observe that
 \[\tww((T(\vec{G},<,\alpha))/\CP_{8n+1}) \leq \tww(G^{\red},<) = t.\]
 It follows that $\tww(T(\vec{G},<,\alpha)) \leq \max(35,t)$ as desired.
\end{proof}

With Lemma \ref{lem:tww-cfi-tournament} in hand, we apply the construction $T(\vec{G},<,\alpha)$ to a $3$-regular base graph $G$ which has tree width linear in the number of vertices, but the twin width of $G^{\red}$ is bounded.
The existence of such graphs has already been observed in \cite{BonnetGKTW22}.
More precisely, the following theorem follows from combining the arguments from \cite[Lemma 5.1]{BonnetGKTW22} and the results from \cite{BiluL06,MarcusSS15}.

\begin{theorem}
 \label{thm:expander-bounded-tww}
 There is a family of $3$-regular graphs $(G_n)_{n \geq 1}$ such that
 \begin{enumerate}
  \item $|V(G_n)| = O(n)$,
  \item $\tww(G_n^{\red}) \leq 6$ (where $G_n^{\red}$ denotes the version of $G_n$ where all edges are turned into red edges), and
  \item $\tw(G_n) \geq n$
 \end{enumerate}
 for every $n \geq 1$.
\end{theorem}

Now, we are ready to give a proof for Theorem \ref{thm:wl-tournament-tww}.

\begin{theorem}
 For every $k \geq 2$ there are non-isomorphic tournaments $T_k$ and $T_k'$ such that
 \begin{enumerate}
  \item $|V(T_k)| = |V(T_k')| = O(k)$,
  \item $\tww(T_k) \leq 35$ and $\tww(T_k') \leq 35$, and
  \item $T_k \simeq_k T_k'$.
 \end{enumerate}
\end{theorem}

\begin{proof}
 Let $k \geq 2$.
 Let $G_{k+1}$ be the $3$-regular graph obtained from Theorem \ref{thm:expander-bounded-tww}.
 Note that $\tw(G_{k+1}) \geq k+1$.

 We also fix an arbitrary orientation $\vec{G}_{k+1}$ of $G_{k+1}$ and let $G_{k+1}^{\red}$ denote the version of $G_{k+1}$ where every edge is replaced by a red edge.
 We have $t \coloneqq \tww(G_{k+1}^{\red}) \leq 6$ by Theorem \ref{thm:expander-bounded-tww}.
 By Lemma \ref{lem:twin-width-order} there is some linear order $<$ on $V(G_{k+1})$ such that $\tww(G_{k+1}^\red,<) \leq t \leq 6$.

 Now fix an arbitrary $u_0 \in V(G_{k+1})$.
 For $p \in \ZZ_3$ we define the mapping $\alpha_p \colon V(G_{k+1}) \to \ZZ_3$ via $\alpha_p(u_0) = p$ and $\alpha_p(w) = 0$ for all $w \in V(G_{k+1}) \setminus \{u_0\}$.
 We define $T_k \coloneqq T(\vec{G}_{k+1},<,\alpha_0)$ and $T_k' \coloneqq T(\vec{G}_{k+1},<,\alpha_1)$.
 We have $|V(T_k)| = |V(T_k')| = 9 \cdot |V(G_{k+1})| = O(k)$.
 Also, $\tww(T_k) \leq 35$ and $\tww(T_k') \leq 35$ by Lemma \ref{lem:tww-cfi-tournament}.
 Finally, $T_k \not\cong T_k'$ by Lemmas \ref{lem:cfi-non-isomorphism} and \ref{lem:cfi-tournament-isomorphism}, and $T_k \simeq_k T_k'$ by Lemma~\ref{lem:cfi-tournament-indistinguishable}.
\end{proof}

\section{Twin Width is Smaller Than Other Widths}
\label{sec:widths}

In this section, we compare twin width with other natural width parameters of tournaments.
If $f,g$ are mappings from (directed) graphs to the natural numbers, we say that $f$ is \emph{functionally smaller} than $g$ on a class $\CC$ of graphs if
for every $k$ there is a $k'$ such that for all graphs $G\in\CC$, if $g(G)\le k$ then $f(G)\le k'$.
We write $f\precsim_{\CC}g$ to denote that $f$ is functionally smaller than $g$ on $\CC$.
We omit the subscript ${}_{\CC}$ if $\CC$ is the class of all digraphs.

Natural width measures for directed graphs are \emph{cut width} \cite{ChudnovskyFS12}, \emph{directed path width}, \emph{directed tree width} \cite{JohnsonRST01}, and \emph{clique width} \cite{CourcelleO00} (see also \cite{Rehs22} for width measures on directed graphs).
On the class of tournaments twin width turns out to be functionally smaller than all of these.
For clique width, it has already been shown in \cite{BonnetKTW22} that twin width is functionally smaller than clique width on undirected graphs; the proof easily extends to arbitrary binary relational structures and hence to tournaments.

We start by giving definitions for the other width measures.

Let $G$ be a digraph.
For a linear order $\leq$ on $V(G)$ and a vertex $v \in V(G)$, we let $S_\leq(v) \coloneqq \{w \in V(G) \mid w \leq v\}$ be the set of all vertices smaller than or equal to $v$ in $\leq$.
Let
\[s_\leq(v) \coloneqq |E_G(S_\leq(v),V(G)\setminus S_\leq(v))|\]
be the number of edges from $S_\le(v)$ to its complement.
The \emph{width} of $\leq$ is $\max_{v\in V(G)} s_\leq(v)$, and the \emph{cut width} $\ctw(G)$ is the minimum over the width of all linear orders of $V(G)$.

A \emph{directed path decomposition} of a digraph $G$ is a mapping $\beta\colon[p] \to 2^{V(G)}$, for some $p\in\NN$, such that
for every vertex $v\in V(G)$ there are $\ell,r \in [p]$ such that
\[v \in\beta(t) \iff \ell \leq t \leq r,\]
and for all edges $(v,w)\in E(G)$ there are $\ell,r \in [p]$ with $\ell \leq r$ such that $v \in \beta(r)$ and $w \in \beta(\ell)$.
The sets $\beta(t)$, $t \in [p]$, are the \emph{bags} of the decomposition.
The \emph{width} of the decomposition is $\max_{t \in [p]}|\beta(t)|-1$,
and the \emph{directed path width} $\dpw(G)$ is the minimum width of a directed path decomposition of $G$.

A digraph $R$ is a \emph{rooted directed tree} if there is a vertex $r_0 \in V(R)$ such that for every $t \in V(R)$ there is a unique directed walk from $r_0$ to $t$.
Note that every rooted directed tree can be obtained from an undirected tree by selecting a root $r_0$ and directing all edges away from the root.
For $t \in V(R)$ we denote by $R_t$ the unique induced subgraph of $R$ rooted at $t$.

Let $G$ be a digraph.
A \emph{directed tree decomposition} of $G$ is a triple $(R,\beta,\gamma)$ where $R$ is a rooted directed tree, $\beta\colon V(R) \to 2^{V(G)}$ and $\gamma\colon E(R) \to 2^{V(G)}$ such that
\begin{enumerate}[label = (D.\arabic*),leftmargin=*]
 \item\label{item:directed-tree-decomposition-1} $\{\beta(t) \mid t \in V(R)\}$ is a partition of $V(G)$, and
 \item\label{item:directed-tree-decomposition-2} for every $(s,t) \in E(R)$ the set $\gamma(s,t)$ is a hitting set for all directed walks that start and end in $\beta(R_t) \coloneqq \bigcup_{t' \in V(R_t)}\beta(t')$ and contain a vertex outside of $\beta(R_t)$.
\end{enumerate}
For $t \in V(R)$ we define $\Gamma(t) \coloneqq \beta(t) \cup \bigcup_{(s,s') \in E(t)} \gamma(s,s')$ where $E(t)$ denotes the set of edges incident to $t$.
The \emph{width} of $(R,\beta,\gamma)$ is defined as
\[\width(R,\beta,\gamma) \coloneqq \max_{t \in V(R)} |\Gamma(t)| - 1.\]
The \emph{directed tree width} $\dtw(G)$ is the minimum width of a directed tree decomposition of $G$.

For the reader's convenience, we include a proof of the following well-known inequalities.

\begin{proposition}
 For all digraphs $G$, it holds that
 \[\dtw(G) \leq \dpw(G) \leq \ctw(G).\]
\end{proposition}

\begin{proof}
 To prove that $\dtw(G) \leq \dpw(G)$, let $\beta\colon [p] \to 2^{V(G)}$ be a directed path decomposition of $G$ of width $k$.
 Let $R$ be the rooted directed tree with vertex set $V(R) \coloneqq [p]$, $E(R) \coloneqq \{(i,i-1) \mid 2 \leq i \leq p\}$, and root $r_0 \coloneqq p$.
 (That is, $R$ is the path $1,\ldots,p$ in reverse order.)
 We let $\beta'(1)\coloneqq \beta(1)$, and for $2 \leq i\leq p$, we let $\beta'(i) \coloneqq \beta(i) \setminus \beta(i-1)$ and $\gamma(i,i-1) \coloneqq \beta(i) \cap \beta(i-1)$.
 It is easy to show that $(R,\beta',\gamma)$ is a directed tree decomposition of $G$ of width $k$.

 To prove that $\dpw(G) \leq \ctw(G)$, let $\leq$ be a linear order of $G$ of width $k$.
 Without loss of generality we may assume that $V(G) = [n]$ and $\leq$ is the natural linear order on $[n]$.
 For every $v \in [n]$, let $\ell_v \coloneqq \max\big(\{v\} \cup N_+(v)\big)$.
 Note that there are at most $k$ vertices $u \leq v$ such that $\ell_u > v$.
 We define $\beta\colon [n] \to 2^{V(G)}$ by $\beta(i) \coloneqq \{v \in V(G) \mid v \leq i \leq \ell_v\}$.
 Note that $|\beta(i)| \leq k+1$.
 It is easy to verify that $\beta$ is a directed path decomposition of $G$ of width at most $k$.
\end{proof}

The proposition implies that $\dtw \precsim \dpw \precsim \ctw$.
The following example shows that $\tww \not\precsim \ctw$ and hence $\tww \not\precsim \dpw$ and $\tww \not\precsim \dtw$ on the class of all digraphs.

\begin{example}
 For every undirected graph $G$, let $\widehat G$ be the digraph with vertex set $V(\widehat G)\coloneqq V(G)\cup E(G)$ and edge set
 \[E(\widehat G) \coloneqq \big\{(e,v),(e,w) \bigmid e = (v,w) \in E(G)\big\}.\]
 Suppose that $V(G) = \{v_1,\dots,v_n\}$ and $E(G) = \{e_1,\dots,e_m\}$.
 Let $\leq$ be the linear order on $V(\widehat G)$ defined by $v_i \leq v_j$ if and only if $i \leq j$,
 $v_i \leq e_j$ for all $i,j$, and $e_i \leq e_j$ if and only if $i \leq j$.
 Since there are only backward edges, the width of this linear order is $0$.
 Hence, $\ctw(\widehat G)=0$.

 However, it is easy to see that the class of all digraphs $\widehat G$ for arbitrary $G$ has unbounded twin width.
\end{example}

In contrast to the last example, it turns out that on the class of tournaments, twin width is functionally smaller than directed tree width.
Actually, this even holds for the larger class of semi-complete graphs.
A digraph $G$ is \emph{semi-complete} if for all distinct $v,w \in V(G)$ at least one of the pairs $(v,w),(w,v)$ is an edge.
Note that every tournament is semi-complete.
We argue that
\begin{equation}
 \label{eq:tww-leq-dtw}
 \tww\precsim_{\CS}\dtw,
\end{equation}
where $\CS$ denotes the class of all semi-complete digraphs.

To prove \eqref{eq:tww-leq-dtw}, we rely on the following result which shows that $\dpw\precsim_{\CS}\dtw$.

\begin{theorem}[{\cite[Proposition 5]{GurskiKRW21}}]
 \label{thm:dpw-vs-dtw}
 Let $G$ be a semi-complete graph.
 Then
 \[\dpw(G) \leq 4(\dtw(G) + 2)^2 + 7(\dtw(G) + 2) - 1.\]
\end{theorem}

Hence, it suffices to show that, on the class of semi-complete digraphs, twin width is functionally smaller than directed path width.

\begin{theorem}
 \label{thm:tww-vs-dpw}
 Let $G$ be a semi-complete graph.
 Then $\tww(G) \leq \dpw(G)$.
\end{theorem}

\begin{proof}
 Let $G$ be a semi-complete graph of directed path width $k$, and let $\beta\colon [p] \to 2^{V(T)}$ be a directed path decomposition of $G$ of width $k$.
 For every $v \in V(G)$, let $r(v) \in [p]$ be maximum such that $v \in \beta(r(v))$.
 Let $v_1,\dots,v_n$ be an enumeration of $V(G)$ such that for all $i < j$ we have $r(v_i) \leq r(v_j)$.

 We define a contraction sequence $\CP_1,\dots,\CP_n$ of $G$ by contracting the vertices from left to right.
 Formally, we define $\CP_i \coloneqq \{\{v_1,\dots,v_i\},\{v_{i+1}\},\dots,\{v_n\}\}$ for all $i \in [n]$.

 Consider the graph $G/\CP_i$.
 Since $P_i \coloneqq \{v_1,\dots,v_i\}$ is the only non-singleton part in $\CP_i$, every red edge in $G/\CP_i$ is incident to $P_i$.
 Suppose $\{P_i,\{v_j\}\}$ is a red edge in $G/\CP_i$ (for some $j > i$).
 Then there are vertices $v,w \in P_i = \{v_1,\dots,v_i\}$ such that $(v,v_j),(v_j,w) \in E(T)$.
 Let $r\coloneqq r(v_i)$.
 Since $v = v_{i'}$ for some $i'\leq i$, we have $r(v) \leq r \leq r(v_j)$.
 By the definition of path decompositions, since $(v,v_j) \in E(T)$ this implies $v_j \in \beta(r)$.

 So if an edge $\{P_i,\{v_j\}\}$ is red in $G/\CP_i$, then $v_j \in \beta(r)$.
 Since we also have $v_i \in \beta(r)$ and $|\beta(r)| \leq k+1$, there can be at most $k$ red edges.
 Hence, the width of $\CP_1,\dots,\CP_n$ is at most $k$.
\end{proof}

Combining Theorems \ref{thm:dpw-vs-dtw} and \ref{thm:tww-vs-dpw}, we obtain the following corollary.

\begin{corollary}
 \label{cor:tww-vs-dtw}
 Let $G$ be a semi-complete graph.
 Then
 \[\tww(G) \leq 4(\dtw(G) + 2)^2 + 7(\dtw(G) + 2) - 1.\]
\end{corollary}

Note that the corollary implies \eqref{eq:tww-leq-dtw}.
We close this section with an example showing that inequality \eqref{eq:tww-leq-dtw} is strict even on tournaments, that is,
\[\dtw\not\precsim_{\CT}\tww\]
where $\CT$ denotes the class of all tournaments.

\begin{example}
 For every $n \geq 1$, let $T_n$ be the tournament obtained by taking a cycle of length $2n+1$ and directing edges along the shorter path.
 Formally, we set $V(T_n) \coloneqq \{0,\dots,2n\}$ and
 \[E(T_n) \coloneqq \big\{ (i,i+j \bmod{2n+1}) \bigmid i \in \{0,\dots,2n\}, j \in \{1,\dots,n\}\big\}.\]
 Then it is easy to show that $\tww(T_n) \leq 1$, but the directed tree width of the family of all $T_n$, $n \geq 1$, is unbounded.
\end{example}

\section{Conclusion}

We prove that the isomorphism problem for classes of tournaments of bounded (or slowly growing) twin width is in polynomial time.
Many algorithmic problems that can be solved efficiently on (classes of) tournaments can also be solved efficiently on (corresponding classes of) semi-complete graphs, that is, directed graphs where for every pair $(v,w)$ of vertices at least one of the pairs $(v,w),(w,v)$ is an edge (see, e.g., \cite{Pilipczuk13}).
Contrary to this, we remark that isomorphism of semi-complete graphs of bounded twin width is GI-complete:
we can reduce isomorphism of oriented graphs to isomorphism of semi-complete graphs by replacing each non-edge by a bidirectional edge.
This reduction preserves twin with.

Classes of tournaments of bounded twin width are precisely the classes that are considered to be structurally sparse.
Formally, these are the classes that are monadically dependent, which means that all set systems definable over the tournaments in such a class have bounded VC dimension.
The most natural set systems definable within a tournament are those consisting of the in-neighbors of the vertices and of the out-neighbors of the vertices.
Bounded twin width implies that the VC dimension of these two set systems is bounded, but the converse does not hold.
It is easy to see that the VC dimensions of the in-neighbors and out-neighbors systems as well as the set system consisting of the mixed neighbors of all edges are within a linear factor of one another.
As a natural next step, we may ask if isomorphism of tournaments where the VC-dimension of these systems is bounded is in polynomial time.

\printbibliography

@inproceedings{GroheN24,
  author       = {Martin Grohe and
                  Daniel Neuen},
  editor       = {Karl Bringmann and
                  Martin Grohe and
                  Gabriele Puppis and
                  Ola Svensson},
  title        = {Isomorphism for Tournaments of Small Twin Width},
  booktitle    = {51st International Colloquium on Automata, Languages, and Programming,
                  {ICALP} 2024, Tallinn, Estonia, July 8-12, 2024},
  series       = {LIPIcs},
  volume       = {297},
  pages        = {78:1--78:20},
  publisher    = {Schloss Dagstuhl - Leibniz-Zentrum f{\"{u}}r Informatik},
  year         = {2024},
  url          = {https://doi.org/10.4230/LIPIcs.ICALP.2024.78},
  doi          = {10.4230/LIPIcs.ICALP.2024.78},
  bibsource    = {dblp computer science bibliography, https://dblp.org}
}

@article{GenietT26,
  author       = {Colin Geniet and
                  St{\'{e}}phan Thomass{\'{e}}},
  title        = {First order logic and twin-width in tournaments and dense oriented
                  graphs},
  journal      = {Eur. J. Comb.},
  volume       = {132},
  pages        = {104247},
  year         = {2026},
  url          = {https://doi.org/10.1016/j.ejc.2025.104247},
  doi          = {10.1016/j.ejc.2025.104247},
  bibsource    = {dblp computer science bibliography, https://dblp.org}
}

@article{Neuen26,
  author       = {Daniel Neuen},
  title        = {Parameterized complexity of graph isomorphism testing},
  journal      = {Comput. Sci. Rev.},
  volume       = {60},
  pages        = {100918},
  year         = {2026},
  url          = {https://doi.org/10.1016/j.cosrev.2026.100918},
  doi          = {10.1016/j.cosrev.2026.100918},
}

@article{ArvindPR25,
  author       = {Vikraman Arvind and
                  Ilia Ponomarenko and
                  Grigory Ryabov},
  title        = {Isomorphism testing of k-spanning tournaments is fixed parameter tractable},
  journal      = {Art Discret. Appl. Math.},
  volume       = {8},
  number       = {2},
  pages        = {2},
  year         = {2025},
  url          = {https://doi.org/10.26493/2590-9770.1712.3ec},
  doi          = {10.26493/2590-9770.1712.3ec},
  bibsource    = {dblp computer science bibliography, https://dblp.org}
}

@article{HlinenyJ25,
  author       = {Petr Hlinen{\'{y}} and
                  Jan Jedelsk{\'{y}}},
  title        = {Twin-Width of Planar Graphs Is at Most 8, and Some Related Bounds},
  journal      = {{SIAM} J. Discret. Math.},
  volume       = {39},
  number       = {4},
  pages        = {2003--2048},
  year         = {2025},
  url          = {https://doi.org/10.1137/23m1623823},
  doi          = {10.1137/23m1623823},
  bibsource    = {dblp computer science bibliography, https://dblp.org}
}

@article{BonnetGKTW24,
  author       = {{\'{E}}douard Bonnet and
                  Colin Geniet and
                  Eun Jung Kim and
                  St{\'{e}}phan Thomass{\'{e}} and
                  R{\'{e}}mi Watrigant},
  title        = {Twin-Width {III:} Max Independent Set, Min Dominating Set, and Coloring},
  journal      = {{SIAM} J. Comput.},
  volume       = {53},
  number       = {5},
  pages        = {1602--1640},
  year         = {2024},
  url          = {https://doi.org/10.1137/21m142188x},
  doi          = {10.1137/21m142188x},
  bibsource    = {dblp computer science bibliography, https://dblp.org}
}

@article{BonnetGMSTT24,
  author       = {{\'{E}}douard Bonnet and
                  Ugo Giocanti and
                  Patrice Ossona de Mendez and
                  Pierre Simon and
                  St{\'{e}}phan Thomass{\'{e}} and
                  Szymon Torunczyk},
  title        = {Twin-Width {IV:} Ordered Graphs and Matrices},
  journal      = {J. {ACM}},
  volume       = {71},
  number       = {3},
  pages        = {21},
  year         = {2024},
  url          = {https://doi.org/10.1145/3651151},
  doi          = {10.1145/3651151},
  bibsource    = {dblp computer science bibliography, https://dblp.org}
}

@inproceedings{Neuen24,
  author       = {Daniel Neuen},
  editor       = {Olaf Beyersdorff and
                  Mamadou Moustapha Kant{\'{e}} and
                  Orna Kupferman and
                  Daniel Lokshtanov},
  title        = {Homomorphism-Distinguishing Closedness for Graphs of Bounded Tree-Width},
  booktitle    = {41st International Symposium on Theoretical Aspects of Computer Science,
                  {STACS} 2024, March 12-14, 2024, Clermont-Ferrand, France},
  series       = {LIPIcs},
  volume       = {289},
  pages        = {53:1--53:12},
  publisher    = {Schloss Dagstuhl - Leibniz-Zentrum f{\"{u}}r Informatik},
  year         = {2024},
  url          = {https://doi.org/10.4230/LIPIcs.STACS.2024.53},
  doi          = {10.4230/LIPIcs.STACS.2024.53},
  bibsource    = {dblp computer science bibliography, https://dblp.org}
}

@article{GroheNS23,
  author       = {Martin Grohe and
                  Daniel Neuen and
                  Pascal Schweitzer},
  title        = {A Faster Isomorphism Test for Graphs of Small Degree},
  journal      = {{SIAM} J. Comput.},
  volume       = {52},
  number       = {6},
  pages        = {S18--1},
  year         = {2023},
  url          = {https://doi.org/10.1137/19m1245293},
  doi          = {10.1137/19m1245293},
  bibsource    = {dblp computer science bibliography, https://dblp.org}
}

@inproceedings{BonnetGMT23,
  author       = {{\'{E}}douard Bonnet and
                  Ugo Giocanti and
                  Patrice Ossona de Mendez and
                  St{\'{e}}phan Thomass{\'{e}}},
  editor       = {Petra Berenbrink and
                  Patricia Bouyer and
                  Anuj Dawar and
                  Mamadou Moustapha Kant{\'{e}}},
  title        = {Twin-Width {V:} Linear Minors, Modular Counting, and Matrix Multiplication},
  booktitle    = {40th International Symposium on Theoretical Aspects of Computer Science,
                  {STACS} 2023, March 7-9, 2023, Hamburg, Germany},
  series       = {LIPIcs},
  volume       = {254},
  pages        = {15:1--15:16},
  publisher    = {Schloss Dagstuhl - Leibniz-Zentrum f{\"{u}}r Informatik},
  year         = {2023},
  url          = {https://doi.org/10.4230/LIPIcs.STACS.2023.15},
  doi          = {10.4230/LIPIcs.STACS.2023.15},
  bibsource    = {dblp computer science bibliography, https://dblp.org}
}

@article{GroheN23,
  author       = {Martin Grohe and
                  Daniel Neuen},
  title        = {Canonisation and Definability for Graphs of Bounded Rank Width},
  journal      = {{ACM} Trans. Comput. Log.},
  volume       = {24},
  number       = {1},
  pages        = {6:1--6:31},
  year         = {2023},
  url          = {https://doi.org/10.1145/3568025},
  doi          = {10.1145/3568025},
  bibsource    = {dblp computer science bibliography, https://dblp.org}
}

@inproceedings{BergeBD22,
  author       = {Pierre Berg{\'{e}} and
                  {\'{E}}douard Bonnet and
                  Hugues D{\'{e}}pr{\'{e}}s},
  editor       = {Mikolaj Bojanczyk and
                  Emanuela Merelli and
                  David P. Woodruff},
  title        = {Deciding Twin-Width at Most 4 Is NP-Complete},
  booktitle    = {49th International Colloquium on Automata, Languages, and Programming,
                  {ICALP} 2022, July 4-8, 2022, Paris, France},
  series       = {LIPIcs},
  volume       = {229},
  pages        = {18:1--18:20},
  publisher    = {Schloss Dagstuhl - Leibniz-Zentrum f{\"{u}}r Informatik},
  year         = {2022},
  url          = {https://doi.org/10.4230/LIPIcs.ICALP.2022.18},
  doi          = {10.4230/LIPIcs.ICALP.2022.18},
  bibsource    = {dblp computer science bibliography, https://dblp.org}
}

@article{BonnetGKTW22,
  author       = {{\'{E}}douard Bonnet and
                  Colin Geniet and
                  Eun Jung Kim and
                  St{\'{e}}phan Thomass{\'{e}} and
                  R{\'{e}}mi Watrigant},
  title        = {Twin-width {II:} small classes},
  journal      = {Comb. Theory},
  volume       = {2},
  number       = {2},
  year         = {2022},
  url          = {https://doi.org/10.5070/c62257876},
  doi          = {10.5070/c62257876},
  bibsource    = {dblp computer science bibliography, https://dblp.org}
}

@inproceedings{BonnetKRT22,
  author       = {{\'{E}}douard Bonnet and
                  Eun Jung Kim and
                  Amadeus Reinald and
                  St{\'{e}}phan Thomass{\'{e}}},
  editor       = {Joseph (Seffi) Naor and
                  Niv Buchbinder},
  title        = {Twin-width {VI:} the lens of contraction sequences},
  booktitle    = {Proceedings of the 2022 {ACM-SIAM} Symposium on Discrete Algorithms,
                  {SODA} 2022, Virtual Conference / Alexandria, VA, USA, January 9 -
                  12, 2022},
  pages        = {1036--1056},
  publisher    = {{SIAM}},
  year         = {2022},
  url          = {https://doi.org/10.1137/1.9781611977073.45},
  doi          = {10.1137/1.9781611977073.45},
  bibsource    = {dblp computer science bibliography, https://dblp.org}
}

@article{BonnetKTW22,
  author       = {{\'{E}}douard Bonnet and
                  Eun Jung Kim and
                  St{\'{e}}phan Thomass{\'{e}} and
                  R{\'{e}}mi Watrigant},
  title        = {Twin-width {I:} Tractable {FO} Model Checking},
  journal      = {J. {ACM}},
  volume       = {69},
  number       = {1},
  pages        = {3:1--3:46},
  year         = {2022},
  url          = {https://doi.org/10.1145/3486655},
  doi          = {10.1145/3486655},
  bibsource    = {dblp computer science bibliography, https://dblp.org}
}

@article{BonnetKRTW22,
  author       = {{\'{E}}douard Bonnet and
                  Eun Jung Kim and
                  Amadeus Reinald and
                  St{\'{e}}phan Thomass{\'{e}} and
                  R{\'{e}}mi Watrigant},
  title        = {Twin-width and Polynomial Kernels},
  journal      = {Algorithmica},
  volume       = {84},
  number       = {11},
  pages        = {3300--3337},
  year         = {2022},
  url          = {https://doi.org/10.1007/s00453-022-00965-5},
  doi          = {10.1007/s00453-022-00965-5},
  bibsource    = {dblp computer science bibliography, https://dblp.org}
}

@inproceedings{GajarskyPT22,
  author       = {Jakub Gajarsk{\'{y}} and
                  Michal Pilipczuk and
                  Szymon Torunczyk},
  editor       = {Christel Baier and
                  Dana Fisman},
  title        = {Stable graphs of bounded twin-width},
  booktitle    = {{LICS} '22: 37th Annual {ACM/IEEE} Symposium on Logic in Computer
                  Science, Haifa, Israel, August 2 - 5, 2022},
  pages        = {39:1--39:12},
  publisher    = {{ACM}},
  year         = {2022},
  url          = {https://doi.org/10.1145/3531130.3533356},
  doi          = {10.1145/3531130.3533356},
  bibsource    = {dblp computer science bibliography, https://dblp.org}
}

@inproceedings{GanianPSSS22,
  author       = {Robert Ganian and
                  Filip Pokr{\'{y}}vka and
                  Andr{\'{e}} Schidler and
                  Kirill Simonov and
                  Stefan Szeider},
  editor       = {Kuldeep S. Meel and
                  Ofer Strichman},
  title        = {Weighted Model Counting with Twin-Width},
  booktitle    = {25th International Conference on Theory and Applications of Satisfiability
                  Testing, {SAT} 2022, August 2-5, 2022, Haifa, Israel},
  series       = {LIPIcs},
  volume       = {236},
  pages        = {15:1--15:17},
  publisher    = {Schloss Dagstuhl - Leibniz-Zentrum f{\"{u}}r Informatik},
  year         = {2022},
  url          = {https://doi.org/10.4230/LIPIcs.SAT.2022.15},
  doi          = {10.4230/LIPIcs.SAT.2022.15},
  bibsource    = {dblp computer science bibliography, https://dblp.org}
}

@inproceedings{NesetrilMS22,
  author       = {Jaroslav Nesetril and
                  Patrice Ossona de Mendez and
                  Sebastian Siebertz},
  editor       = {Florin Manea and
                  Alex Simpson},
  title        = {Structural Properties of the First-Order Transduction Quasiorder},
  booktitle    = {30th {EACSL} Annual Conference on Computer Science Logic, {CSL} 2022,
                  February 14-19, 2022, G{\"{o}}ttingen, Germany (Virtual Conference)},
  series       = {LIPIcs},
  volume       = {216},
  pages        = {31:1--31:16},
  publisher    = {Schloss Dagstuhl - Leibniz-Zentrum f{\"{u}}r Informatik},
  year         = {2022},
  url          = {https://doi.org/10.4230/LIPIcs.CSL.2022.31},
  doi          = {10.4230/LIPIcs.CSL.2022.31},
  bibsource    = {dblp computer science bibliography, https://dblp.org}
}

@inproceedings{Thomasse22,
  author       = {St{\'{e}}phan Thomass{\'{e}}},
  editor       = {Mikolaj Bojanczyk and
                  Emanuela Merelli and
                  David P. Woodruff},
  title        = {A Brief Tour in Twin-Width (Invited Talk)},
  booktitle    = {49th International Colloquium on Automata, Languages, and Programming,
                  {ICALP} 2022, July 4-8, 2022, Paris, France},
  series       = {LIPIcs},
  volume       = {229},
  pages        = {6:1--6:5},
  publisher    = {Schloss Dagstuhl - Leibniz-Zentrum f{\"{u}}r Informatik},
  year         = {2022},
  url          = {https://doi.org/10.4230/LIPIcs.ICALP.2022.6},
  doi          = {10.4230/LIPIcs.ICALP.2022.6},
  bibsource    = {dblp computer science bibliography, https://dblp.org}
}

@article{Neuen22,
  author       = {Daniel Neuen},
  title        = {Hypergraph Isomorphism for Groups with Restricted Composition Factors},
  journal      = {{ACM} Trans. Algorithms},
  volume       = {18},
  number       = {3},
  pages        = {27:1--27:50},
  year         = {2022},
  url          = {https://doi.org/10.1145/3527667},
  doi          = {10.1145/3527667},
  bibsource    = {dblp computer science bibliography, https://dblp.org}
}

@inproceedings{GroheN21,
  author       = {Martin Grohe and
                  Daniel Neuen},
  editor       = {Konrad K. Dabrowski and
                  Maximilien Gadouleau and
                  Nicholas Georgiou and
                  Matthew Johnson and
                  George B. Mertzios and
                  Dani{\"{e}}l Paulusma},
  title        = {Recent advances on the graph isomorphism problem},
  booktitle    = {Surveys in Combinatorics, 2021: Invited lectures from the 28th British
                  Combinatorial Conference, Durham, UK, July 5-9, 2021},
  pages        = {187--234},
  publisher    = {Cambridge University Press},
  year         = {2021},
  url          = {https://doi.org/10.1017/9781009036214.006},
  doi          = {10.1017/9781009036214.006},
  bibsource    = {dblp computer science bibliography, https://dblp.org}
}

@inproceedings{GurskiKRW21,
  author       = {Frank Gurski and
                  Dominique Komander and
                  Carolin Rehs and
                  Sebastian Wiederrecht},
  editor       = {Ding{-}Zhu Du and
                  Donglei Du and
                  Chenchen Wu and
                  Dachuan Xu},
  title        = {Directed Width Parameters on Semicomplete Digraphs},
  booktitle    = {Combinatorial Optimization and Applications - 15th International Conference,
                  {COCOA} 2021, Tianjin, China, December 17-19, 2021, Proceedings},
  series       = {Lecture Notes in Computer Science},
  volume       = {13135},
  pages        = {615--628},
  publisher    = {Springer},
  year         = {2021},
  url          = {https://doi.org/10.1007/978-3-030-92681-6\_48},
  doi          = {10.1007/978-3-030-92681-6\_48},
  bibsource    = {dblp computer science bibliography, https://dblp.org}
}

@article{GajarskyKNMPST20,
  author       = {Jakub Gajarsk{\'{y}} and
                  Stephan Kreutzer and
                  Jaroslav Nesetril and
                  Patrice Ossona de Mendez and
                  Michal Pilipczuk and
                  Sebastian Siebertz and
                  Szymon Torunczyk},
  title        = {First-Order Interpretations of Bounded Expansion Classes},
  journal      = {{ACM} Trans. Comput. Log.},
  volume       = {21},
  number       = {4},
  pages        = {29:1--29:41},
  year         = {2020},
  url          = {https://doi.org/10.1145/3382093},
  doi          = {10.1145/3382093},
  bibsource    = {dblp computer science bibliography, https://dblp.org}
}

@article{Kiefer20,
  author       = {Sandra Kiefer},
  title        = {The {W}eisfeiler-{L}eman algorithm: an exploration of its power},
  journal      = {{ACM} {SIGLOG} News},
  volume       = {7},
  number       = {3},
  pages        = {5--27},
  year         = {2020},
  url          = {https://doi.org/10.1145/3436980.3436982},
  doi          = {10.1145/3436980.3436982},
  bibsource    = {dblp computer science bibliography, https://dblp.org}
}

@article{ChudnovskySS19,
  author       = {Maria Chudnovsky and
                  Alex Scott and
                  Paul D. Seymour},
  title        = {Disjoint paths in unions of tournaments},
  journal      = {J. Comb. Theory, Ser. {B}},
  volume       = {135},
  pages        = {238--255},
  year         = {2019},
  url          = {https://doi.org/10.1016/j.jctb.2018.08.007},
  doi          = {10.1016/j.jctb.2018.08.007},
  bibsource    = {dblp computer science bibliography, https://dblp.org}
}

@article{FominP19,
  author       = {Fedor V. Fomin and
                  Michal Pilipczuk},
  title        = {On width measures and topological problems on semi-complete digraphs},
  journal      = {J. Comb. Theory, Ser. {B}},
  volume       = {138},
  pages        = {78--165},
  year         = {2019},
  url          = {https://doi.org/10.1016/j.jctb.2019.01.006},
  doi          = {10.1016/j.jctb.2019.01.006},
  bibsource    = {dblp computer science bibliography, https://dblp.org}
}

@article{ChudnovskyKLST18,
  author       = {Maria Chudnovsky and
                  Ringi Kim and
                  Chun{-}Hung Liu and
                  Paul D. Seymour and
                  St{\'{e}}phan Thomass{\'{e}}},
  title        = {Domination in tournaments},
  journal      = {J. Comb. Theory, Ser. {B}},
  volume       = {130},
  pages        = {98--113},
  year         = {2018},
  url          = {https://doi.org/10.1016/j.jctb.2017.10.001},
  doi          = {10.1016/j.jctb.2017.10.001},
  bibsource    = {dblp computer science bibliography, https://dblp.org}
}

@book{Grohe17,
  author    = {Martin Grohe},
  title     = {Descriptive Complexity, Canonisation, and Definable Graph Structure
               Theory},
  series    = {Lecture Notes in Logic},
  volume    = {47},
  publisher = {Cambridge University Press},
  year      = {2017},
  url       = {https://doi.org/10.1017/9781139028868},
  doi       = {10.1017/9781139028868},
  isbn      = {9781139028868},
  bibsource = {dblp computer science bibliography, https://dblp.org}
}

@inproceedings{Schweitzer17,
  author       = {Pascal Schweitzer},
  editor       = {Ioannis Chatzigiannakis and
                  Piotr Indyk and
                  Fabian Kuhn and
                  Anca Muscholl},
  title        = {A Polynomial-Time Randomized Reduction from Tournament Isomorphism
                  to Tournament Asymmetry},
  booktitle    = {44th International Colloquium on Automata, Languages, and Programming,
                  {ICALP} 2017, July 10-14, 2017, Warsaw, Poland},
  series       = {LIPIcs},
  volume       = {80},
  pages        = {66:1--66:14},
  publisher    = {Schloss Dagstuhl - Leibniz-Zentrum f{\"{u}}r Informatik},
  year         = {2017},
  url          = {https://doi.org/10.4230/LIPIcs.ICALP.2017.66},
  doi          = {10.4230/LIPIcs.ICALP.2017.66},
  bibsource    = {dblp computer science bibliography, https://dblp.org}
}

@inproceedings{Babai16,
  author       = {L{\'{a}}szl{\'{o}} Babai},
  editor       = {Daniel Wichs and
                  Yishay Mansour},
  title        = {Graph isomorphism in quasipolynomial time [extended abstract]},
  booktitle    = {Proceedings of the 48th Annual {ACM} {SIGACT} Symposium on Theory
                  of Computing, {STOC} 2016, Cambridge, MA, USA, June 18-21, 2016},
  pages        = {684--697},
  publisher    = {{ACM}},
  year         = {2016},
  url          = {https://doi.org/10.1145/2897518.2897542},
  doi          = {10.1145/2897518.2897542},
  bibsource    = {dblp computer science bibliography, https://dblp.org}
}

@article{MarcusSS15,
    AUTHOR = {Marcus, Adam W. and Spielman, Daniel A. and Srivastava,
              Nikhil},
     TITLE = {Interlacing families {I}: {B}ipartite {R}amanujan graphs of
              all degrees},
   JOURNAL = {Ann. of Math. (2)},
  FJOURNAL = {Annals of Mathematics. Second Series},
    VOLUME = {182},
      YEAR = {2015},
    NUMBER = {1},
     PAGES = {307--325},
      ISSN = {0003-486X,1939-8980},
       DOI = {10.4007/annals.2015.182.1.7},
       URL = {https://doi.org/10.4007/annals.2015.182.1.7},
}

@article{ChudnovskyFS12,
  author       = {Maria Chudnovsky and
                  Alexandra Ovetsky Fradkin and
                  Paul D. Seymour},
  title        = {Tournament immersion and cutwidth},
  journal      = {J. Comb. Theory, Ser. {B}},
  volume       = {102},
  number       = {1},
  pages        = {93--101},
  year         = {2012},
  url          = {https://doi.org/10.1016/j.jctb.2011.05.001},
  doi          = {10.1016/j.jctb.2011.05.001},
  bibsource    = {dblp computer science bibliography, https://dblp.org}
}

@article{ChudnovskyS11,
  author       = {Maria Chudnovsky and
                  Paul D. Seymour},
  title        = {A well-quasi-order for tournaments},
  journal      = {J. Comb. Theory, Ser. {B}},
  volume       = {101},
  number       = {1},
  pages        = {47--53},
  year         = {2011},
  url          = {https://doi.org/10.1016/j.jctb.2010.10.003},
  doi          = {10.1016/j.jctb.2010.10.003},
  bibsource    = {dblp computer science bibliography, https://dblp.org}
}

@inproceedings{DawarR07,
  author       = {Anuj Dawar and
                  David Richerby},
  editor       = {Jacques Duparc and
                  Thomas A. Henzinger},
  title        = {The Power of Counting Logics on Restricted Classes of Finite Structures},
  booktitle    = {Computer Science Logic, 21st International Workshop, {CSL} 2007, 16th
                  Annual Conference of the EACSL, Lausanne, Switzerland, September 11-15,
                  2007, Proceedings},
  series       = {Lecture Notes in Computer Science},
  volume       = {4646},
  pages        = {84--98},
  publisher    = {Springer},
  year         = {2007},
  url          = {https://doi.org/10.1007/978-3-540-74915-8\_10},
  doi          = {10.1007/978-3-540-74915-8\_10},
  bibsource    = {dblp computer science bibliography, https://dblp.org}
}

@article{BiluL06,
  author       = {Yonatan Bilu and
                  Nathan Linial},
  title        = {Lifts, Discrepancy and Nearly Optimal Spectral Gap},
  journal      = {Comb.},
  volume       = {26},
  number       = {5},
  pages        = {495--519},
  year         = {2006},
  url          = {https://doi.org/10.1007/s00493-006-0029-7},
  doi          = {10.1007/s00493-006-0029-7},
  bibsource    = {dblp computer science bibliography, https://dblp.org}
}

@article{JohnsonRST01,
  author       = {Thor Johnson and
                  Neil Robertson and
                  Paul D. Seymour and
                  Robin Thomas},
  title        = {Directed Tree-Width},
  journal      = {J. Comb. Theory, Ser. {B}},
  volume       = {82},
  number       = {1},
  pages        = {138--154},
  year         = {2001},
  url          = {https://doi.org/10.1006/jctb.2000.2031},
  doi          = {10.1006/jctb.2000.2031},
  bibsource    = {dblp computer science bibliography, https://dblp.org}
}

@article{CourcelleO00,
  author       = {Bruno Courcelle and
                  Stephan Olariu},
  title        = {Upper bounds to the clique width of graphs},
  journal      = {Discret. Appl. Math.},
  volume       = {101},
  number       = {1-3},
  pages        = {77--114},
  year         = {2000},
  url          = {https://doi.org/10.1016/S0166-218X(99)00184-5},
  doi          = {10.1016/S0166-218X(99)00184-5},
  bibsource    = {dblp computer science bibliography, https://dblp.org}
}

@article{Hella96,
  author       = {Lauri Hella},
  title        = {Logical Hierarchies in {PTIME}},
  journal      = {Inf. Comput.},
  volume       = {129},
  number       = {1},
  pages        = {1--19},
  year         = {1996},
  url          = {https://doi.org/10.1006/inco.1996.0070},
  doi          = {10.1006/inco.1996.0070},
  bibsource    = {dblp computer science bibliography, https://dblp.org}
}

@article{SeymourT93,
  author    = {Paul D. Seymour and
               Robin Thomas},
  title     = {Graph Searching and a Min-Max Theorem for Tree-Width},
  journal   = {J. Comb. Theory, Ser. {B}},
  volume    = {58},
  number    = {1},
  pages     = {22--33},
  year      = {1993},
  url       = {https://doi.org/10.1006/jctb.1993.1027},
  doi       = {10.1006/jctb.1993.1027},
  bibsource = {dblp computer science bibliography, https://dblp.org}
}

@article{CaiFI92,
  author       = {Jin{-}yi Cai and
                  Martin F{\"{u}}rer and
                  Neil Immerman},
  title        = {An optimal lower bound on the number of variables for graph identification},
  journal      = {Comb.},
  volume       = {12},
  number       = {4},
  pages        = {389--410},
  year         = {1992},
  url          = {https://doi.org/10.1007/BF01305232},
  doi          = {10.1007/BF01305232},
  bibsource    = {dblp computer science bibliography, https://dblp.org}
}

@incollection{ImmermanL90,
  author    = {Immerman, Neil and Lander, Eric},
  editor    = {Selman, Alan L.},
  title     = {Describing Graphs: A First-Order Approach to Graph Canonization},
  booktitle = {Complexity Theory Retrospective: In Honor of Juris Hartmanis on the Occasion of His Sixtieth Birthday, July 5, 1988},
  year      = {1990},
  publisher = {Springer New York},
  address   = {New York, NY},
  pages     = {59--81},
  isbn      = {978-1-4612-4478-3},
  doi       = {10.1007/978-1-4612-4478-3_5}
}

@inproceedings{BabaiL83,
  author       = {L{\'{a}}szl{\'{o}} Babai and
                  Eugene M. Luks},
  editor       = {David S. Johnson and
                  Ronald Fagin and
                  Michael L. Fredman and
                  David Harel and
                  Richard M. Karp and
                  Nancy A. Lynch and
                  Christos H. Papadimitriou and
                  Ronald L. Rivest and
                  Walter L. Ruzzo and
                  Joel I. Seiferas},
  title        = {Canonical Labeling of Graphs},
  booktitle    = {Proceedings of the 15th Annual {ACM} Symposium on Theory of Computing,
                  25-27 April, 1983, Boston, Massachusetts, {USA}},
  pages        = {171--183},
  publisher    = {{ACM}},
  year         = {1983},
  url          = {https://doi.org/10.1145/800061.808746},
  doi          = {10.1145/800061.808746},
  bibsource    = {dblp computer science bibliography, https://dblp.org}
}

@article{Miller83,
  author       = {Gary L. Miller},
  title        = {Isomorphism of Graphs Which are Pairwise k-separable},
  journal      = {Inf. Control.},
  volume       = {56},
  number       = {1/2},
  pages        = {21--33},
  year         = {1983},
  url          = {https://doi.org/10.1016/S0019-9958(83)80048-5},
  doi          = {10.1016/S0019-9958(83)80048-5},
  bibsource    = {dblp computer science bibliography, https://dblp.org}
}

@article{Luks82,
  author       = {Eugene M. Luks},
  title        = {Isomorphism of Graphs of Bounded Valence can be Tested in Polynomial
                  Time},
  journal      = {J. Comput. Syst. Sci.},
  volume       = {25},
  number       = {1},
  pages        = {42--65},
  year         = {1982},
  url          = {https://doi.org/10.1016/0022-0000(82)90009-5},
  doi          = {10.1016/0022-0000(82)90009-5},
  bibsource    = {dblp computer science bibliography, https://dblp.org}
}

@article{NesetrilO16,
    AUTHOR = {Jaroslav Nesetril and
              Patrice Ossona de Mendez},
     TITLE = {Structural sparsity},
   JOURNAL = {Russian Math. Surveys},
  FJOURNAL = {Russian Mathematical Surveys},
    VOLUME = {71},
      YEAR = {2016},
    NUMBER = {1},
     PAGES = {79--107},
       DOI = {10.4213/rm9688},
       URL = {https://doi.org/10.4213/rm9688},
}

@book{Seress03,
    AUTHOR = {Seress, \'{A}kos},
     TITLE = {Permutation Group Algorithms},
    SERIES = {Cambridge Tracts in Mathematics},
    VOLUME = {152},
 PUBLISHER = {Cambridge University Press, Cambridge},
      YEAR = {2003},
     PAGES = {x+264},
      ISBN = {0-521-66103-X},
       DOI = {10.1017/CBO9780511546549},
       URL = {https://doi.org/10.1017/CBO9780511546549},
}

@book{Rotman99,
    AUTHOR = {Rotman, Joseph J.},
     TITLE = {An Introduction to the Theory of Groups},
    SERIES = {Graduate Texts in Mathematics},
    VOLUME = {148},
   EDITION = {Fourth},
 PUBLISHER = {Springer-Verlag, New York},
      YEAR = {1995},
     PAGES = {xvi+513},
      ISBN = {0-387-94285-8},
       DOI = {10.1007/978-1-4612-4176-8},
       URL = {https://doi.org/10.1007/978-1-4612-4176-8},
}

@book{DixonM96,
    AUTHOR = {Dixon, John D. and Mortimer, Brian},
     TITLE = {Permutation Groups},
    SERIES = {Graduate Texts in Mathematics},
    VOLUME = {163},
 PUBLISHER = {Springer-Verlag, New York},
      YEAR = {1996},
     PAGES = {xii+346},
      ISBN = {0-387-94599-7},
       DOI = {10.1007/978-1-4612-0731-3},
       URL = {https://doi.org/10.1007/978-1-4612-0731-3},
}

@article{Ponomarenko92,
    AUTHOR = {Ponomarenko, Ilia N.},
     TITLE = {Polynomial time algorithms for recognizing and isomorphism
              testing of cyclic tournaments},
   JOURNAL = {Acta Appl. Math.},
  FJOURNAL = {Acta Applicandae Mathematicae},
    VOLUME = {29},
      YEAR = {1992},
    NUMBER = {1-2},
     PAGES = {139--160},
       DOI = {10.1007/BF00053383},
       URL = {https://doi.org/10.1007/BF00053383},
}

@article{FeitT63,
    AUTHOR = {Feit, Walter and Thompson, John G.},
     TITLE = {Solvability of groups of odd order},
   JOURNAL = {Pacific J. Math.},
  FJOURNAL = {Pacific Journal of Mathematics},
    VOLUME = {13},
      YEAR = {1963},
     PAGES = {775--1029}
}

@phdthesis{Rehs22,
  author       = {Carolin Rehs},
  title        = {Comparing and Computing Parameters for Directed Graphs},
  school       = {University of D{\"{u}}sseldorf, Germany},
  year         = {2022},
}

@phdthesis{Pilipczuk13,
  author = {Michal Pilipczuk},
  title  = {Tournaments and Optimality: New Results in Parameterized Complexity},
  school = {University of Bergen},
  year   = {2013}
}

@book{Weisfeiler76,
  author    = {Weisfeiler, Boris},
  title     = {On Construction and Identification of Graphs},
  volume    = {558},
  series    = {Lecture Notes in Mathematics},
  year      = {1976},
  publisher = {Springer-Verlag},
  url       = {https://doi.org/10.1007/BFb0089374},
  doi       = {10.1007/BFb0089374},
}

@article{WeisfeilerL68,
  Author  = {Boris Weisfeiler and Andrei Leman},
  Journal = {NTI, Series 2},
  Note    = {English translation by Grigory Ryabov available at \url{https://www.iti.zcu.cz/wl2018/pdf/wl_paper_translation.pdf}},
  Title   = {The reduction of a graph to canonical form and the algebra which appears therein},
  Year    = {1968}
}
\end{document}